\documentclass[draft,onecolumn]{IEEEtran}
\ifCLASSINFOpdf
\else
\fi
\usepackage[cmex10]{amsmath}
\usepackage{amsthm}
\usepackage{bbm}
\usepackage{color}
\usepackage{cite}

\usepackage{amssymb}
\usepackage{stmaryrd}
\usepackage{epic,eepic}

\usepackage{epsf}           
\usepackage{graphicx}       
\usepackage{epsfig}         
\usepackage{pstricks}
\usepackage{pstricks-add}
\usepackage{pst-plot}
\usepackage{dsfont}
\usepackage{hypernat}
\usepackage{hyperref}

\newtheorem{theorem}{Theorem}
\newtheorem{remark}[theorem]{Remark}
\newtheorem{definition}[theorem]{Definition}
\newtheorem{proposition}[theorem]{Proposition}
\newtheorem{lemma}[theorem]{Lemma}

\newcommand{\mc}{\mathcal}

\newcommand{\PP}{\mathcal{P}}

\newcommand{\QQ}{\mathcal{Q}}
\newcommand{\R}{\mathbb{R}}

\newcommand{\N}{\mathbb{N}}
\newcommand{\E}{\mathbb{E}}
\newcommand{\Q}{\mathbb{Q}}

\newcommand{\sign}{\mathrm{sign}}

\newcommand\norm[1]{\left\lVert#1\right\rVert}
\newcommand{\eps}{{\varepsilon}}
\newcommand{\set}[1]{\mathcal{#1}}

\newcommand{\Aen}{\set{A}_\eps^{(n)}}
\newcommand{\Ben}{\set{B}_\eps^{(n)}}
\newcommand{\Sen}{\set{S}_\eps^{(n)}}

\definecolor{darkgreen}{rgb}{0.0,0.6,0}

\usepackage[nomessages]{fp}
\newcommand\const[3][6]{%
    \edef\temporary{round(#3}%
    \expandafter\FPeval\csname#2\expandafter\endcsname
        \expandafter{\temporary:#1)}%
        \pstVerb{/#2 \csname#2\endcsname\space def}%
}

\newcommand{\FigMMSEThird}[2]{

\const{Pa}{0.5 * (#1 - 2*#2 - (#1*(#1- 4*#2))^0.5)}
\const{Pb}{0.5 * (#1 - 2*#2 + (#1*(#1- 4*#2))^0.5)}

\const{MMSEzero}{#1 * #2 / (#1 + #2)}

\const{MaxOrd}{\MMSEzero * 1.1}
\const{MaxAbs}{#1 * 1.1}

\const{MinOrd}{-0.1 * \MaxOrd}
\const{MinAbs}{-0.1 * \MaxAbs}

\const{XUNIT}{1/\MaxAbs * 7}
\const{YUNIT}{1/\MaxOrd * 6}

\const{AbsLin}{#1- #2}
\const{OrdLin}{#2/#1 * (#1 - #2)}

\def\MMSEsa{ #1 sqrt x sqrt neg add dup mul #2 mul #1 sqrt x sqrt neg add dup mul #2 add div}
\def\MMSElin{#1 x neg add  #2 neg add #2 mul #1 div}

\const{MMSEPa}{#2 /(2*#1) * (#1 + (#1*(#1-4*#2))^0.5)}
\const{MMSEPb}{#2 /(2*#1) * (#1- (#1*(#1-4*#2))^0.5)}

\const{MinOrd}{-0.07}
\const{OrdConstant}{0.75}

\begin{figure}[!h]
\begin{center}
\psset{xunit= 60 cm,yunit= 500  cm}				
\begin{pspicture}(-0.02,-0.0004)(0.17,0.0105)			
\fileplot[linecolor=darkgreen]{data/DataTwoPoint_#1_#2.dat}
\fileplot[linecolor=red]{data/Data_DPC_lin_#1_#2.dat}
\fileplot[linecolor=orange]{data/Data_ICw1_#1_#2.dat}
\psline{->}(0,0)(0,\MaxOrd)
\psline{->}(\MinAbs,0)(\MaxAbs,0)
\psplot[plotpoints=200,linecolor=blue]{0}{#1}{ \MMSEsa}
\psplot[plotpoints=200,linecolor=brown]{0}{\AbsLin}{\MMSElin}

\rput[r](-0.2,\MinOrd){$0$}
\rput[u](\MaxAbs,\MinOrd){$P$}
\rput[u](#1,\MinOrd){$Q$}

\rput[u](#1,\MinOrd){$Q$}

\rput[l](\Pa,\MinOrd){$P_1$}
\rput[u](\Pb,\MinOrd){$P_2$}
\psline[linestyle=dashed](0,0.0067)(0.0369,0)
\psdots(0.026135,0.00195)

\rput[l](0.052,0.008){$Q=#1$}
\rput[l](0.052,0.007){$N=#2$}

\rput[r](-0.003,0.009){$\mathsf{MMSE}$}
\rput[r](-0.002,-0.0007){$0$}
\rput[r](0.007,-0.0008){$P_1$}
\rput[r](0.082,-0.0008){$P_2$}
\rput[r](0.105,-0.0007){$P$}

\psline[linestyle=dotted](\Pa,0)(\Pa,\OrdLin)
\psline[linestyle=dotted](\Pb,0)(\Pb,\OrdLin)
\psline[linestyle=dotted](0,\MMSEPa)(\AbsLin,\MMSEPa)
\psline[linestyle=dotted](0,\MMSEPb)(\AbsLin,\MMSEPb)

\psdots(\AbsLin,0)(#1,0)(\Pa,0)(\Pb,0)
\psdot(0,\MMSEzero)
\psdot(0,\MMSEPa)
\psdot(0,\MMSEPb)

\psframe(0.087,0.0087)(0.172,0.0033)
\psline[linecolor=blue](0.090,0.008)(0.095,0.008)
\rput[l](0.097,0.008){$\mathsf{MMSE}_{\ell}(P)$ in  \eqref{eq:MMSElinear}}
\psline[linecolor=brown](0.090,0.007)(0.095,0.007)
\rput[l](0.097,0.007){$ \frac{N\cdot(Q-N-P)}{Q} $ in \eqref{eq:SolutionContinuous}}
\psline[linecolor=darkgreen](0.090,0.006)(0.095,0.006)
\rput[l](0.097,0.006){$(P_{\mathsf{two}}, \mathsf{MMSE}_{\mathsf{two}})$ in \eqref{eq:Power_W}, \eqref{eq:MMSE_W} }
\psline[linecolor=red](0.090,0.005)(0.095,0.005)
\rput[l](0.097,0.005){$\mathsf{MMSE}_{\mathsf{lin+dpc}}(P)$ in \eqref{eq:MMSE_lin_DPC}}
\psline[linecolor=orange](0.090,0.004)(0.095,0.004)
\rput[l](0.097,0.004){$\mathsf{MMSE}_{\mathsf{coord}}(P) $ in \eqref{eq:MMSE_coord}}

\end{pspicture}
\caption{Comparison of the proposed estimation costs $\mathsf{MMSE}_{\mathsf{G}}(P) $ in \eqref{eq:SolutionContinuous}, and $\mathsf{MMSE}_{\mathsf{coord}}(P) $ in \eqref{eq:MMSE_coord}, with Witsenhausen two point strategy $(P_{\mathsf{two}}, \mathsf{MMSE}_{\mathsf{two}})$ in \eqref{eq:Power_W}, \eqref{eq:MMSE_W}, and Grover and Sahai's combination of the linear scheme with DPC scheme $\mathsf{MMSE}_{\mathsf{lin+dpc}}(P)$ in \eqref{eq:MMSE_lin_DPC}. }\label{fig:OptimalMMSE_#1_#2_bis}
\end{center}
\end{figure}}

\begin{document}
%
\title{Power-Estimation Trade-off of
Vector-valued Witsenhausen Counterexample 
with Causal Decoder
\thanks{
The authors gratefully acknowledge the financial support of SRV ENSEA for visits at KTH in Stockholm in 2017 and 2019, and at ETIS in Cergy in 2018. This research has been conducted as part of the Labex MME-DII (ANR11-LBX-0023-01). Part of the research has been supported by Swedish Research Council (VR) under grant 2020-03884. This article was presented in part at the 2018 IEEE Allerton conference \cite{MO18ocdv}, at the 2019 IEEE Information Theory Workshop (ITW) \cite{OM19ccwc}, and at the 2021 IEEE International Symposium on Information Theory (ISIT) \cite{MO21crve}.

Maël Le Treust is with ERMINE team of IRISA UMR 6074, 35000 Rennes, France (e-mail: mael.le-treust@cnrs.fr).

Tobias Oechtering is with Division of Information Science and Engineering, KTH, Stockholm, Sweden (e-mail: oech@kth.se).
}
}
%
%
%

\author{\IEEEauthorblockN{Ma\"{e}l Le~Treust,~\IEEEmembership{Member,~IEEE, }}
\and
\IEEEauthorblockN{Tobias J. Oechtering,~\IEEEmembership{Senior Member,~IEEE}
}}

%
%

\markboth{Journal of \LaTeX\ Class Files,~Vol.~14, No.~8, August~2015}%
{Shell \MakeLowercase{\textit{et al.}}: Bare Demo of IEEEtran.cls for IEEE Journals}
%



\maketitle


\begin{abstract}
The vector-valued extension of the famous Witsenhausen counterexample setup is studied where the encoder, i.e. the first decision maker, non-causally knows and encodes the i.i.d. state sequence and the decoder, i.e. the second decision maker, causally estimates the interim state. The coding scheme is transferred from the finite alphabet coordination problem for which it is proved to be optimal. The extension to the Gaussian setup is based on a non-standard weak typicality approach and requires a careful average estimation error analysis since the interim state is estimated by the decoder. We provide a single-letter expression that characterizes the optimal trade-off between the Witsenhausen power cost and estimation cost. The two auxiliary random variables improve the communication with the decoder, while performing the dual role of the channel input, which also controls the state of the system. Interestingly, we show that a pair of discrete and continuous auxiliary random variables, outperforms both Witsenhausen two point strategy and the best affine policies. The optimal choice of random variables remains unknown.
\end{abstract}



\section{Introduction}
Distributed decision-making systems arise in engineering problems where decentralized agents choose actions based on locally available information as to minimize a common cost function. The information at each agent is either locally observed or received from other agents. Since the 
process of sharing information comes with a cost, agents usually do not have access to the whole information available at all agents. The design of optimal decision strategies for such distributed stochastic networks with non-classical information structures is a long-standing difficult problem. The famous counterexample of Witsenhausen introduced in 1968 in \cite{W68acis} showed that non-linear strategies can outperform the best linear strategy. Until today, the setup serves as important study object to develop a better understanding on the impact of the information structure on the optimal decision strategy design problem \cite{YB2013sncs}.  This has significantly helped to highlight the fundamental difficulty that actions serve two purposes, a control purpose affecting the system state and a communication purpose providing information to other agents. More generally, this problem is also referred to a team decision problem for which the existence of the 
optimal solution has been studied in \cite{GYBL15oteo}.


Although Witsenhausen refuted with his simple two-point counterexample the assertion that a linear policy would be optimal in such a Gaussian setting, the optimal non-linear policy remains unknown. Many researcher have approached the optimization problem with various methods. In the last decade for instance it has been approached with numerical optimization methods \cite{TT17alsa},\cite{KGOS11iscc}, where the latter is based on an iterative joint source-channel coding approach. An asymptotically optimal approximation technique has been presented in \cite{SYL17fmaa}. Analytically,
using results from optimal transport theory, it has been shown in \cite{WV11wcav} 
that the optimal decision strategy is a strictly increasing unbounded piece-wise real analytic function with a real analytic left inverse. More necessary conditions have been derived in \cite{MH15ofam}, by analyzing an equivalent optimization problem on the space of square-integrable quantile functions. In \cite{KC15aoat}, a framework to find the optimal joint distribution for certain stochastic control problems with non-classical information structure has been presented. It is argued that non-classical information structure leads to non-convex problems which makes them difficult to solve. It is therefore proposed to solve a convex relaxation problem with the same objective but larger convex feasible region, that is also a solution to the non-convex problem.  The data-processing inequality plays a critical role in the construction of 
such relaxations convexifying the problem, therewith generalizing the approach by Bansal and Basar \cite{BB87stwn}. For the Witsenhausen problem, this relaxation approach results in a lower bound.

Another approach 
is to consider a multi-letter version of the problem. In a series of works, \cite{GS10wcaa,GPS13aost,GWS15ieat} to mention a few, Grover et al. studied the setup where the decision makers have access to a sequence of observations enabling block-coding strategies. This allowed them to transfer advanced coding techniques \cite{GP80cfcw,KSC08sa,SS09iewr,SCCK05ccas} to the vector-valued Witsenhausen counterexample problem. In one of their  last works \cite{GWS15ieat}, which also provides a good literature overview, they extended the concept of {\it dual} (role of) control to {\it triple} roles by adding an explicit communication task to the problem, highlighting the fundamental tension among the tasks. Most of the results were derived for finite alphabet setup where the concept of strong typicality provides the Markov Lemma, see \cite[Lemma~12.1]{EK11nit}. A rigorous extension of the coding scheme to the Gaussian case has been done by Grover and Sahai in \cite{GS10wcaa}. In \cite{GPS13aost}, approximately lattice-based optimal solutions were obtained for the finite-length vector case. Recently, improved asymptotic bounds have been found in \cite{MBS17ipbf}, using a new vector quantization scheme. Choudhuri and Mitra characterized the optimal power-distortion trade-off for the vector-valued Wistenhausen problem in \cite{CM12owct}, which relies on the coding scheme by Grover et al. in \cite{GS10wcaa} combining linear coding and Costa's dirty-paper-coding \cite{Cwodp}.
Much less work has been done considering the vector-valued Witsenhausen problem with causal processing although several coding techniques have been extended to the causal case, e.g. causal state communication in \cite{CKM13csc} or estimation with a helper in \cite{CSW13ewah}.

In \cite{CZ2011cuic}, Cuff and Zhao considered empirical coordination for a \emph{cascade of controllers} that act on its observed signals where the empirical coordination criterion is a probabilistic statement on the statistics of the joint sequences. In particular, they point out that given a reward function, then the optimal average reward can be found by optimizing over the coordination set. 
An extension to more abstract alphabets has been done in \cite{R13epts}, introducing a new definition of typical sequences and deriving properties using the Glivenko-Cantelli Theorem. Originally, the problem of coordination between agents was introduced in \cite{GHN06ouoc} by Gossner et al., for a two-player team game with asymmetric information. The authors also discussed the case of noisy observations, which relates to the case of noisy communication channels. In \cite{CPC10cc}, the concept has been generalized and the notion of coordination capacity (region) has been introduced, which can be used to characterize the joint behavior of distributed nodes, given communication constraints.
In particular, results for simple multi-source settings considering empirical and strong coordination have been obtained. In \cite{CS2011hcnf}, Cuff and Schieler investigated the case where the action of terminal one has also to be coordinated with the state and terminal two's action. Interestingly, the achievability proof relies on a hybrid coding strategy which can be seen as the multi-letter extension of the best Witsenhausen counterexample decision strategy. 


In \cite{MT14cbcs,MT15ecwt,MT15ctec,LLW2015cisd} and \cite{LLW2015cpia}, empirical coordination capacity results for two terminal settings with side information have been derived, by considering state-dependent channels as well as causal and non-causal encoding and decoding. In \cite{MT15ecwt}, necessary and sufficient conditions for the non-causal encoding and decoding case for a {\it cascade setting}  have been presented, 
which includes the results of the lossless decoding case with correlated source and state presented in \cite{MT14cbcs}. Optimal results have been obtained for the perfect channel case with two-sided side information and the case with independent source and channel. Further, optimal results have been presented for causal/non-causal encoding and decoding including sketches of the achievability and converse proofs, while full proofs were provided in \cite{MT15ctec} and \cite{MT17jeco}. 
In contrast, the authors of \cite{LLW2015cisd} characterize optimal conditions for a setting where both terminals (a.k.a. agents) provide a channel input whose output, that also depends on the system state, is observed by terminal two only. Next, they consider the special case without any channel input from terminal two for which they characterize, in \cite[Theorem 3]{LLW2015cisd}, the optimal solution for non-causal encoder and causal decoder.



An improved understanding of the fundamental distributed decision making problem is of great value due to its wide applications. For instance, Larrousse et al. applied in \cite{LAL2014icit,LLB18cidn} the coordination approach to a two agents distributed power allocation problem  where only one agent is knowledgable about a state and informs the second agent through its actions. The idea of (state) communication through actions is usually known as {\it dual control} where control actions have a second purpose. In \cite{ADLL15icit}, coordination in a two-agent setting with common average payoff function where each agent can control only one variable is considered. Such payoff function includes the Witsenhausen cost function as special case. The authors assumed standard Borel spaces to justify the transfer of  coding results, from finite alphabets to continuous alphabets.

In this work, we study fundamental bounds of a vector-valued Witsenhausen counterexample setup with a non-causal control strategy at the first decision maker and a causal control strategy at the second decision maker. We characterize the fundamental trade-off between the achievable power cost and estimation cost. The achievablity and converse follow from the corresponding empirical coordination coding scheme \cite{MT15ecwt}. The coding scheme in this work is based on the concept of weak typicality with an extension to avoid the use of the Markov Lemma, as done in \cite{VOS18hiwp}. Moreover, we extend average estimation error analysis, as done in \cite{W78trdf}, to deal with an estimate of an interim state and not an i.i.d. source. Since we use weak typicality, the result readily applies to continuous alphabets considering Gaussian distributed memoryless source and noise, as in the Witsenhausen counterexample setup. We next study the power-estimation cost trade-off of control schemes considering different choices of auxiliary random variables. In particular, we consider the class of auxiliary random variables that are jointly Gaussian and characterize the optimal estimation cost for given power. Restricted to Gaussian distributions, the optimal control policy is {\it state  contraction}, which is memoryless and linear. We next derive the estimation cost considering a hybrid case with a discrete and a Gaussian distributed auxiliary random variable, where the discrete random variable encapsulates  the sign of the intermediate state. In the following, numerical discussion we show that this hybrid strategy outperforms the best linear policy,  
likewise Witsenhausen two-point control strategy. The fact that the Witsenhausen two-point strategy and its hybrid coding strategies outperform the best linear strategy in the Gaussian case considering causal decoder implies that the observation from Witsenhausen counterexample extends to block-coding schemes, which might be an interesting observation considering other source-channel coding problems.

The model is presented in Section~\ref{sec:model}, the characterization of the optimal trade-off between power and estimation costs is stated in Section~\ref{sec:MainResult}. In  Section~\ref{sec:ParticularInterest}, we propose two control schemes and we compare their performances to the results from the literature. Our control scheme that involves a discrete auxiliary random variable and a continuous auxiliary random variable outperforms both Witsenhausen two point strategy, and the best affine policies. The conclusion is stated in Section~\ref{sec:conclusion} and the proofs of the results are in Appendices~\ref{sec:achiev}-\ref{sec:ProofPropDPC}.

\IEEEpeerreviewmaketitle

\begin{figure}[!ht]
\begin{center}
\psset{xunit=1cm,yunit=1cm}
\begin{pspicture}(0.5,-0.5)(12,2)
\psframe(2,0)(3,1)

\pscircle(4.2,0.5){0.2}
\rput(4.2,0.5){$+$}
\pscircle(5.6,0.5){0.2}
\rput(5.6,0.5){$+$}
\psframe(6.8,0)(7.8,1)

\psdots[linewidth=2pt](1,1.5)(5.6,1.5)
\rput[u](2.4,1.8){$X_0^{n}\sim \mathcal{N}(0,Q\mathbb{I})$}
\psline[linewidth=1pt]{->}(1,1.5)(1,0.5)(2,0.5)
\psline[linewidth=1pt]{->}(1,1.5)(4.2,1.5)(4.2,0.7)
\psline[linewidth=1pt]{->}(3,0.5)(4,0.5)
\psline[linewidth=1pt]{->}(4.4,0.5)(5.4,0.5)
\psline[linewidth=1pt]{->}(4.9,0.5)(4.9,-0.5)(8.8,-0.5)
\psline[linewidth=1pt]{->}(5.6,1.5)(5.6,0.7)
\psline[linewidth=1pt]{->}(5.8,0.5)(6.8,0.5)
\psline[linewidth=1pt]{->}(7.8,0.5)(8.8,0.5)

\rput[u](1.5,0.8){$X_0^n$}
\rput[u](3.5,0.8){$U_1^n$}
\rput[u](4.9,0.8){$X_1^{n}$}
\rput[u](6.3,0.8){$Y_1^{t}$}
\rput[u](6,1.8){$Z_1^{n}\sim \mathcal{N}(0,N\mathbb{I})$}
\rput[u](8.3,0.8){$U_{2,t}$}
\rput[u](8.3,-0.2){$X_{1,t}$}
\rput(2.5,0.5){$C_1$}
\rput(7.3,0.5){$C_2$}

\end{pspicture}
\caption{The state and the channel noise are drawn according to the i.i.d. Gaussian distributions $X_0^{n}\sim \mathcal{N}(0,Q\mathbb{I})$ and $Z_1^{n}\sim \mathcal{N}(0,N\mathbb{I})$.}\label{fig:model}
\end{center}
\end{figure}
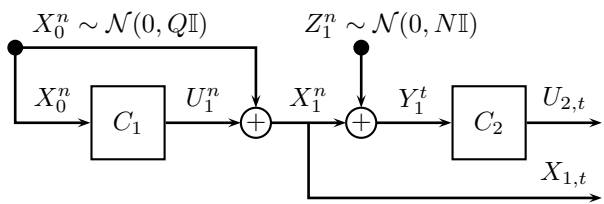


\section{System Model}\label{sec:model}

Throughout this paper, calligraphic fonts, e.g. $\mc{X}_1$, stand for sets, capital letters, e.g. $X_1$, denote random variables while lowercase letters, e.g. $x_1 \in\mc{X}_1$ denote realizations. Sequences of length $n\in\N^{\star} = \N \setminus\{0\}$ of random variables and realizations are denoted respectively by $X_1^n = (X_{1,1}, \ldots, X_{1,t},\ldots, X_{1,n})$ and $x_1^n = (x_{1,1},\ldots, x_{1,t},\ldots,x_{1,n})$, for all $t\in\{1,\ldots,n\}$.

We consider the vector-valued Witsenhausen setup depicted in Fig.~\ref{fig:model}. The notations $\mc{X}_0$, $\mc{U}_1$, $\mc{X}_1$, $\mc{Y}_1$, $\mc{U}_2$ stand for the sets of states, channel inputs, interim states, channel outputs, receiver outputs, that are all equal to the real line $\R$. For $n\in \N^{\star}$, the $n$-time Cartesian product of sets is denoted by $\mc{X}_0^n$. The sequences of states and channel noises are drawn independently according to the i.i.d. Gaussian distributions $X_0^{n}\sim \mathcal{N}(0,Q\mathbb{I})$ and $Z_1^{n}\sim \mathcal{N}(0,N\mathbb{I})$ with $\min(Q,N)>0$, where $\mathbb{I}\in\mathbb{R}^{n\times n}$ denotes the identity matrix. We denote by $X_1$ the \emph{interim state} and $Y_1$ the output of the noisy channel, defined by
\begin{align}
X_1 =& X_0 + U_1 \qquad\qquad\qquad\;\quad\quad  \text{ with }\;X_0 \sim\; \mc{N}(0,Q),\label{eq:Gaussian1} \\
Y_1 =& X_1 + Z_1=X_0 + U_1+ Z_1 \quad \text{ with }\;Z_1 \sim\; \mc{N}(0,N).\label{eq:Gaussian3}
\end{align}

We denote by $\PP_{X_0} = \mc{N}(0,Q)$ the Gaussian probability distribution of the state random variable $X_0$, and we denote by $\PP_{X_1Y_1|X_0U_1}$ the conditional probability distribution corresponding to \eqref{eq:Gaussian1} and \eqref{eq:Gaussian3}.

\begin{definition}\label{def:Code}
For $n\in\N^{\star}$, a ``control design'' with non-causal encoder and causal decoder is a tuple of stochastic functions $c=(f,\{g_t\}_{t\in\{1,\ldots,n\}})$ defined by
\begin{align}
f : \mc{X}_0^n  \longrightarrow   \mc{U}_1^n,\quad
g_t :  \mc{Y}_1^{t}  \longrightarrow \mc{U}_2 ,   \;\forall t \in \{1, \ldots,n\}, \label{eq:Code}
\end{align}
which induces a distribution over the sequences of symbols given by
\begin{align}
&\bigg(\prod_{t=1}^n \PP_{X_{0,t}}\bigg)f_{U_1^n|X_0^n}\bigg(\prod_{i=t}^n \PP_{X_{1,t}Y_{1,t}|X_{0,t}U_{1,t}} \bigg)
\bigg( \prod_{t=1}^n g_{U_{2,t}|Y_1^{t}}\bigg).\label{eq:Distribution}
\end{align}

We denote by $\mc{C}_{\mathsf{d}}(n)$ the set of control designs with non-causal encoder and causal decoder.
\end{definition}

The Witsenhausen counterexample in \cite{W68acis}, investigates the trade-off between two cost functions, a power cost for the channel input $U_1$, and a decoder estimation cost of the interim state $X_1$. We evaluate these two costs by considering their respective averages over the sequences of symbols.

\begin{definition}
We define the $n$-stage costs associated with control design $c\in \mc{C}_{\mathsf{d}}(n)$ by
\begin{align}
\gamma^n_{\mathsf{p}}(c) =&
\begin{cases}
 \E\Big[\frac{1}{n}\sum_{t=1}^n {U_{1,t}}^2\Big] &\text{if it exists,}\\
+\infty &\text{otherwise,}
\end{cases}\\
\gamma^n_{\mathsf{s}}(c) =&\begin{cases}
 \E\Big[ \frac{1}{n}\sum_{t=1}^n (X_{1,t} - U_{2,t})^2\Big] &\text{if it exists,}\\
+\infty &\text{otherwise.}
\end{cases}
\end{align}

The pair of costs $(P,S)\in\R^2$ is achievable if for all $\varepsilon>0$, there exists $\bar{n}\in\N^{\star}$ such that for all $n\geq \bar{n}$, there exists a control design $c\in\mc{C}_{\mathsf{d}}(n)$ such that
\begin{align}
\Big| P - \gamma^n_{\mathsf{p}}(c) \Big| + 
\Big| S - \gamma^n_{\mathsf{s}}(c) \Big| \leq \varepsilon.\label{eq:EqualConstraint}
\end{align}
\end{definition}

The goal is to characterize the set of achievable pair of costs $(P,S)\in\R^2$, which we call the Witsenhausen costs.

\section{Coding Result}\label{sec:MainResult}

In this section, we extend the coordination coding result of \cite[Theorem 4]{MT17jeco} to the case where state and channel noise are Gaussian random variables, other random variables are real-valued, and we consider the power and estimation cost functions of the vector-valued Witsenhausen counterexample. We provide a characterization of the  achievable pairs of costs $(P,S)$.

\begin{theorem}\label{theo:MainResult}
The pair of Witsenhausen costs $(P,S)$ is achievable if and only if  there exists a joint probability distribution that decomposes according to 
\begin{align}
&\PP_{X_0} \QQ_{U_1W_1W_2|X_0} \PP_{X_1Y_1|X_0U_1} \QQ_{U_2|W_2Y_1},\label{eq:TargetDistribution0}
\end{align}
 such that 
 \begin{align}
& I(W_1;Y_1|W_2) - I(W_1,W_2;X_0)\geq 0,\label{eq:TheoIC}\\
&P = \E_{\QQ}\big[U_{1}^2\big],\qquad
S = \E_{\QQ}\big[(X_{1} - U_{2})^2\big],\label{eq:TheoCost}
\end{align}
where $W_1$ and $W_2$ are discrete or continuous auxiliary random variables.
\end{theorem}

The achievability proof of Theorem \ref{theo:MainResult} relies on a block-Markov coding scheme and an adequate notion of weak typicality, inspired from the techniques of \cite{VOS18hiwp}. The achievability and converse proofs are stated in App.\ref{sec:achiev}. Entropy and mutual information are defined using the Radon-Nikodym derivative, which  is briefly recapitulated in the preliminaries in App.\ref{sec:achiev_prelim}.

\begin{remark}
The probability distributions in \eqref{eq:TargetDistribution0} satisfy the Markov chains
\begin{align}
\begin{cases}
(X_1,Y_1)  -\!\!\!\!\minuso\!\!\!\!- (X_0 ,U_1) -\!\!\!\!\minuso\!\!\!\!-  (W_1,W_2) ,\\ 
U_2 -\!\!\!\!\minuso\!\!\!\!- (Y_1 , W_2 ) -\!\!\!\!\minuso\!\!\!\!-  (X_0,X_1,U_1, W_1 ).
\end{cases}\label{eq:MarkovChains}
\end{align}
The causality condition prevents the controller $C_2$ to recover $W_1$ which induces the second Markov chain of \eqref{eq:MarkovChains}. The first and second Markov chains are induced by the network topology.
\end{remark}

\begin{remark}
The information constraint \eqref{eq:TheoIC} reformulates as 
\begin{align}
I(W_1;Y_1,W_2) - I(W_1;X_0,W_2)\geq I(X_0;W_2).
\end{align}
The terms $I(X_0;W_2)$ corresponds to the quantization of the state $X_0$ via the auxiliary random variable $W_2$. The expression $I(W_1;Y_1,W_2) - I(W_1;X_0,W_2)$ stands for the capacity of a two-sided state dependent channel where the encoder observes $(X_0,W_2)$ and the decoder observes $W_2$. Intuitively, $W_1$ is used to tune the state-dependent channel so as to increase its capacity, as in \cite{GP80cfcw}, in order to refine the quantization of the state $X_0$, via the auxiliary random variable $W_2$. 
\end{remark}

In order to investigate the region of achievable pairs of Witsenhausen costs $(P,S)$, we focus on its boundary. We fix the power cost to some parameter $P\geq0$, and we investigate the optimal estimation cost at the decoder. 

\begin{definition}
Given a power cost parameter $P \geq0$, the optimal estimation cost $\mathsf{MMSE}(P)$ is the solution of the optimization problem 
\begin{align}
\mathsf{MMSE}(P) =& \inf_{\QQ \in \Q(P)}{\E_{\QQ}\big[(X_{1} - U_{2})^2\big]},\label{eq:MMSEContinuous} \\
\Q(P) =& \bigg\{ \big( \QQ_{U_1W_1W_2|X_0}, \QQ_{U_2|W_2Y_1}\big) \;\text{ s.t. }\; P = E_{\QQ}\big[U_{1}^2\big], \nonumber \\
&\qquad I(W_1;Y_1|W_2) - I(W_1,W_2;X_0)\geq0   \bigg\}.\label{eq:InformationConstraintContinuous}
\end{align}
\end{definition}

The notation $\mathsf{MMSE}(P)$ in \eqref{eq:MMSEContinuous} recalls that the decoder estimation cost is the Minimum Mean Square Error estimation cost. For such objective, the decoder optimal decision policy is well known and given by the conditional expectation stated in the following lemma.

\begin{lemma}\label{lemma:OptimalU2}
Given a power cost parameter $P \geq0$, the optimal estimation cost $\mathsf{MMSE}(P)$ satisfies
\begin{align}
\mathsf{MMSE}(P) =& \inf_{\QQ \in \Q_{\mathsf{c}}(P)}{\E_{\QQ}\Big[(X_{1} - \E[X_1|W_2,Y_1])^2\Big|W_2,Y_1\Big]},\label{eq:MMSEContinuousE} \\
 \Q_{\mathsf{c}}(P) =& \bigg\{  \QQ_{U_1W_1W_2|X_0} \;\text{ s.t. }\; P = E_{\QQ}\big[U_{1}^2\big], \nonumber \\
&\qquad  I(W_1;Y_1|W_2) - I(W_1,W_2;X_0)\geq0   \bigg\}.\label{eq:InformationConstraintContinuousE}
\end{align}
\end{lemma}

\begin{proof}[Lemma \ref{lemma:OptimalU2}]
For all probability distribution $\QQ_{X_1W_2Y_1}$, the random variable $U_2=\E[X_1|W_2,Y_1]$ minimizes $\E_{\QQ}\big[(X_{1} - U_2)^2\big]$.
\end{proof}

Note that the remaining optimization problem \eqref{eq:MMSEContinuousE} is difficult to solve since the domain is the set of real-valued distributions. In the next section, we investigate this optimization problem by considering additional assumptions that restrict the set of conditional distributions $\QQ_{U_1W_1W_2|X_0}$.

\section{Control schemes of particular interest}\label{sec:ParticularInterest}

We restrict our attention to specific choices for the  auxiliary random variables $W_1$ and $W_2$. We consider both $W_1$ and $W_2$ are Gaussian in Sec.~\ref{sec:GaussianW1W2}, whereas in Sec.~\ref{sec:GaussianW1dDiscreteW2}, the random variable $W_2$ is discrete and $W_1$ is Gaussian.

\subsection{Gaussian auxiliary random variables $W_1$ and $W_2$}\label{sec:GaussianW1W2}

We focus our attention to the class of jointly Gaussian  random variables $(X_0,U_1,W_1,W_2,X_1,Y_1,U_2)$.

\begin{definition}
Given a power cost parameter $P \geq0$, we define the optimal estimation cost obtained with Gaussian random variables
\begin{align}
\mathsf{MMSE}_{\mathsf{G}}(P)  =& \inf_{\QQ \in \Q_{\mathsf{G}}(P)}{\E_{\QQ}\Big[(X_{1} - \E[X_1|W_2,Y_1])^2\Big|W_2,Y_1\Big]},\label{eq:MMSEContinuousEG} \\
\Q_{\mathsf{G}}(P) =& \bigg\{  \QQ_{U_1W_1W_2|X_0} \;\text{ is conditionally Gaussian and }\; P = E_{\QQ}\big[U_{1}^2\big], \nonumber \\
&\qquad \qquad \quad I(W_1;Y_1|W_2) - I(W_1,W_2;X_0)\geq0   \bigg\}.\label{eq:InformationConstraintContinuousEG}
\end{align}
\end{definition}

Note that if $(X_0,U_1,W_2)$ are Gaussian, then $\E[X_1|W_2,Y_1]$ is also Gaussian.

\begin{definition}
Given a power cost parameter $P \geq0$, we consider the linear scheme defined by 
\begin{align}
U_{1,\ell}(P)=& 
\begin{cases}
-\sqrt{\frac{P}{Q}} X_0&\text{ if } P\in[0,Q],\\
-X_0+\sqrt{P-Q} \qquad\qquad &\text{ otherwise.}
\end{cases}\label{eq:MMSElinearStrat}
\end{align}
The linear estimation cost function given by
\begin{align}
\mathsf{MMSE}_{\ell}(P)=& 
\begin{cases}
\frac{ \big(\sqrt{Q}-  \sqrt{P} \big)^2 \cdot N}{ \big(\sqrt{Q}-  \sqrt{P} \big)^2 + N}\quad &\text{ if } P\in[0,Q],\\
0 \qquad\qquad &\text{ otherwise.}
\end{cases}\label{eq:MMSElinear}
\end{align}
\end{definition}

Note that if $P\geq Q$, the interim state $X_1$ can be canceled and the offset $\sqrt{P-Q}$ 
is only included to meet the power constraint with equality, as in \eqref{eq:EqualConstraint}. In the linear scheme $U_{1,\ell}(P)$, the channel input is used to cancel the state $X_0$. The next lemma is a reformulation of \cite[Lemma 11]{W68acis}, which shows that $\mathsf{MMSE}_{\ell}(P)$ is obtained by using the best linear scheme.

\begin{lemma}\label{lemma:LinearScheme}
We consider the linear strategy
\begin{align}
U_1 = a \cdot X_0 +b,
\end{align}
with parameters $(a,b)\in \R^2$, such that to match the power cost constraint $\E\big[{U_1}^2\big]= a^2 Q + b^2 = P$. The optimal estimation cost is given by
\begin{align}
\inf_{(a,b)\in R^2,\atop a^2 Q + b^2 = P}\E\Big[(X_1 - \E[X_1|Y_1])^2)\Big|Y_1\Big] 
=\mathsf{MMSE}_{\ell}(P),
\end{align}
which is achieved by the strategy $U_{1,\ell}(P)$ defined in \eqref{eq:MMSElinearStrat}.
\end{lemma}
For the sake of clarity, we also provide the proof of Lemma \ref{lemma:LinearScheme}, in App. \ref{sec:ProofLemmaLinear}.

\begin{theorem}
\label{theo:CharacterizationContinuousRV}
Suppose that $Q>4N$, we define the parameters
\begin{align}
P_1 =& \frac{1}{2} \Big(Q - 2N - \sqrt{Q\cdot(Q-4N)}\Big),\\
P_2 =& \frac{1}{2} \Big(Q - 2N +\sqrt{Q\cdot(Q-4N)}\Big).
\end{align}

The optimal estimation cost obtained with Gaussian random variables is given by
\begin{align}
\mathsf{MMSE}_{\mathsf{G}}(P) =&
\begin{cases}
\frac{N\cdot(Q-N-P)}{Q} & \text{if } Q> 4N \text{ and } P\in[P_1,P_2],\\
\mathsf{MMSE}_{\ell}(P) & \text{otherwise. }
\end{cases}
\label{eq:SolutionContinuous}
\end{align}
\end{theorem}

The proof of Theorem \ref{theo:CharacterizationContinuousRV} is stated in App. \ref{sec:ProofTheoGaussian}. The estimation cost in  \eqref{eq:SolutionContinuous} can be obtained by using, either a time sharing strategy between the two linear schemes $U_{1,\ell}(P_1)$ and $U_{1,\ell}(P_2)$, when $Q> 4N$ and $P\in[P_1,P_2]$, and otherwise with the linear scheme $U_{1,\ell}(P)$. This result shows that, under the Gaussian assumption,  memoryless policies are optimal so that these policies are also optimal for the original scalar Witsenhausen counterexample setup restricted to Gaussian random variables. However, as pointed out by Witsenhausen  in \cite{W68acis}, the Gaussian assumption is a strong restriction in the original scalar model which induces control designs that are generally not optimal.


\subsection{Gaussian auxiliary random variable $W_1$ and discrete $W_2$}\label{sec:GaussianW1dDiscreteW2}

In this section, we assume that $P\leq Q$ and we consider that $W_2$ is a discrete auxiliary random variable, equal to the sign of the interim random variable $X_1$, 
\begin{align}
W_2 =& \sign(X_1).\label{eq:discreteW2}
\end{align}
We assume that the random variables $(X_0,U_1)$ are centered jointly Gaussian, distributed according to  $\mathcal{N}(0,K)$, with covariance matrix 
\begin{align}
K = 
\begin{pmatrix}
Q &\rho\sqrt{PQ}\\
\rho\sqrt{PQ}& P
\end{pmatrix},\label{eq:CorrelationMatrix}
\end{align}
depending on the correlation parameter $\rho\in[-1,1]$. 

Given a correlation parameter $\rho\in[-1,1]$, we reformulate the pair of correlated Gaussian random variables $(X_0,U_1)$ into a pair of independent Gaussian random variables $(\tilde{S} ,\tilde{X} )$ such that the sum is preserved, i.e. $X_0+U_1 =  \tilde{S} +\tilde{X} $, and the auxiliary channel input $\tilde{X}$ is independent of the auxiliary channel state $\tilde{S}$. Since $(X_0,U_1)\sim\mathcal{N}(0,K)$, we have $U_1 = \rho\sqrt{\frac{P}{Q}} X_0+\tilde{X}$ with $\tilde{X}\sim \mathcal{N}(0,P(1-\rho^2))$ and $\tilde{X} \perp X_0$. Therefore, we introduce two Gaussian random variables $\tilde{X}\sim \mathcal{N}(0,P(1-\rho^2))$  and $\tilde{S} \sim \mathcal{N}(0,(\sqrt{Q}+\rho \sqrt{P})^2)$ such that 
\begin{align}
\tilde{S} =&\;\frac{\sqrt{Q}+\rho \sqrt{P}}{\sqrt{Q}}\cdot X_0,\label{eq:Stilde}\\
\tilde{X} \perp &\; (\tilde{S},X_0),\\
X_1 =&\; X_0+U_1 = \tilde{X}+ \tilde{S}.
\end{align}

Then, the state-dependent channel reformulates
\begin{align}
Y_1 =& X_0+U_1 + Z = \tilde{X}+ \tilde{S} +Z,
\end{align}
for which Costa's auxiliary random variable for Dirty Paper Coding (DPC), see \cite{Cwodp}, writes
\begin{align}
W_1 =& \tilde{X} + \alpha \tilde{S} ,\quad \text{ with }\quad \alpha = \frac{P(1-\rho^2)}{P(1-\rho^2)+N}.\label{eq:alpha}
\end{align}

By combining \eqref{eq:Stilde} and \eqref{eq:alpha}, we reformulate the auxiliary random variable
\begin{align}
W_1 =& \tilde{X} + \frac{P(1-\rho^2)(\sqrt{Q}+\rho \sqrt{P})}{(P(1-\rho^2)+N)\sqrt{Q}} \cdot X_0,\qquad \text{ where } X_0\perp \tilde{X}\sim \mathcal{N}(0,P(1-\rho^2))\label{eq:discreteW1}.
\end{align}

Note that the correlation parameter $\rho\in[-1,1]$ is a free parameter that we use to minimize the decoder estimation cost.

\begin{definition}
Given $P\geq0$, we consider the auxiliary random variables $(W_1,W_2)$ defined by \eqref{eq:discreteW1} and \eqref{eq:discreteW2}, the optimal estimation cost is defined by
\begin{align}
\mathsf{MMSE}_{\mathsf{coord}}(P) = &\min_{\rho\in[-1,1]} \E\Big[\Big(X_1 - \E[X_1|W_2,Y_1]\Big)^2\Big| W_2,Y_1\Big],\\
&\quad\text{s.t. }\quad I(W_1;Y_1|W_2) - I(W_1;X_0|W_2) \geq I(X_0;W_2).\label{eq:IC_MMSE_coord}
\end{align}
\end{definition}
In the information constraint \eqref{eq:IC_MMSE_coord}, the quantization rate $I(X_0;W_2)$ must be smaller than the state dependent channel capacity $I(W_1;Y_1|W_2) - I(W_1;X_0|W_2)$.

\begin{proposition}\label{prop:IC_W2sign}
Given $P\geq0$, we consider the auxiliary random variables $(W_1,W_2)$ defined by \eqref{eq:discreteW1} and \eqref{eq:discreteW2} and we use the change of variable $T = P + Q+ 2\rho \sqrt{PQ}$. We have
\begin{align}
\mathsf{MMSE}_{\mathsf{coord}}(P)  =& \min_{\rho\in[-1,1]}\frac{T N}{T + N} \cdot \Bigg(1 -  \frac{2}{\sqrt{T+N}} \cdot \frac{1}{2\pi} \int\frac{  \phi\Big( y_1 \cdot \sqrt{\frac{2 T+N}{N( T+N)}}\Big) }{  \Phi\Big(y_1 \cdot \sqrt{\frac{T}{N( T+N)}}\Big)}  dy_1  \Bigg),\label{eq:MMSE_coord}\\
& \quad \text{s.t. } \quad  \frac12 \log_2\Bigg(1+\frac{P(1-\rho^2)}{N}\Bigg)  -  \Psi\Bigg(\sqrt{\frac{T}{N}}\Bigg)\nonumber\\
&\qquad \quad + \Psi\Bigg( \sqrt{\frac{T(\sqrt{Q}+\rho\sqrt{P})^2N+P(1-\rho^2)(T+N)^2}{(\sqrt{Q}+\rho\sqrt{P})^2N^2}}\Bigg)\geq 1,\label{eq:IC_discreteCount}
 \end{align}
where the entropy reduction function $\Psi:\R \to [0,1]$ is defined by 
\begin{align}
\Psi(\alpha) =&\int2\Phi\big(\alpha \cdot x\big)\cdot\log_2\Big(2\Phi\big(\alpha \cdot x\big)\Big) \frac{1}{\sqrt{2\pi}}\exp\Big(-\frac{x^2}{2}\Big)dx,\label{eq:EntropyReductionFunction}
\end{align}
and $\Phi(x) = \frac{1}{\sqrt{2\pi}}\int_{-\infty}^{x} \exp\big(-\frac{t^2}{2}\big) \textrm{dt} $, is the Gaussian cumulative distribution function.
\end{proposition}
The proof of Prop. \ref{prop:IC_W2sign} is stated in App. \ref{sec:ProofPropW2discrete}. The first term of \eqref{eq:IC_discreteCount} corresponds to the capacity of a state-dependent channel with channel input power $P(1-\rho^2)$ and noise variance $N$. The entropy reduction function  $\Psi(\alpha)$ corresponds to the entropy penalty term of skew normal distribution with the skewness factor $\alpha\in \R$.

\begin{figure}[!h]
\begin{center}
\psset{xunit= 0.4cm,yunit= 1cm}			
\begin{pspicture}(-10,-0.1)(10,1.4) 
\fileplot[linecolor=blue]{data/Data_TermEntropySkew.dat}
\psline{->}(0,-0.0005)(0,1.4)
\psline{->}(-10,0)(10,0)
\psline[linestyle=dotted](-10,1)(10,1)
\psline[linestyle=dotted](-5,0)(-5,1)
\psline[linestyle=dotted](5,0)(5,1)
\rput[c](9,1.2){$\Psi(\alpha)$}
\rput[c](9,-0.2){$\alpha$}
\rput[c](0,-0.2){$0$}
\rput[r](-0.15,1.2){$1$}
\rput[c](5,-0.2){$5$}
\rput[c](-5,-0.2){$-5$}
\end{pspicture}
\caption{Entropy reduction function $\Psi(\alpha)$ defined in \eqref{eq:EntropyReductionFunction}.}\label{fig:EntropyReductionFunction}
\end{center}
\end{figure}
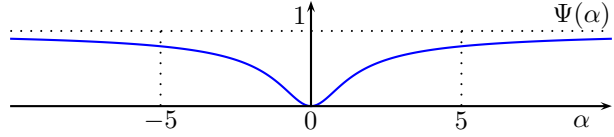

\subsection{Numerical Results}\label{sec:NumericalResults}

We compare the performances of the control schemes of Sec.~\ref{sec:GaussianW1W2} and Sec.~\ref{sec:GaussianW1dDiscreteW2}, to the Witsenhausen two-point strategy \cite[Sec.~5]{W68acis}. In \cite{GS10wcaa}, Grover and Sahai investigate a vector-version of the Witsenhausen counterexample in which the decoder is non-causal and implements Costa's DPC scheme, see \cite{Cwodp}, for a specific channel state. We compare the control schemes of Sec.~\ref{sec:GaussianW1W2} and Sec.~\ref{sec:GaussianW1dDiscreteW2} to the DPC based scheme of  \cite{GS10wcaa}.

\subsubsection{Witsenhausen two-point strategy}\label{sec:TwoPointStrategy}
\begin{proposition}[Two-point strategy]\label{prop:TwoPointStrategy}
For some parameter $a\geq0$, Witsenhausen two-point strategy is defined by 
\begin{align}
U_1 =& a \cdot \textrm{sign}\big(X_0\big) - X_0.
\end{align}
The power and estimation costs are given by
\begin{align}
P_{\mathsf{two}}(a) =&Q + a\Big(a-2\sqrt{\frac{2Q}{\pi}}\Big),\label{eq:Power_W}\\
\mathsf{MMSE}_{\mathsf{two}}(a) =& \sqrt{\frac{2\pi}{N}}  a^2\phi\bigg(\frac{a}{\sqrt{N}}\bigg)\int  \frac{\phi\big(\frac{y}{\sqrt{N}}\big)}{\cosh\big(\frac{a y }{N}\big)}dy,\label{eq:MMSE_W}
\end{align}
where $\phi(x)=  \frac{1}{\sqrt{2 \pi}}\exp\big(-\frac{x^2}{2}\big)$ and the optimal receiver's strategy is given by $\E[X_1|Y=y] =a \tanh\big(\frac{a y }{N}\big)$.
\end{proposition}

For the sake of clarity, we also provide the proof of Proposition \ref{prop:TwoPointStrategy}, in App. \ref{sec:ProofPropTwoPoint}. By letting $N=1$ and $a=\sqrt{Q}$, we recover the equations in the proof of \cite[Theorem 2]{W68acis}.  Note that the function $P_{\mathsf{two}}(a)$ decreases over the interval $\Big[0,\sqrt{\frac{2Q}{\pi}}\Big]$, where it reaches the minimal value $Q\big(1 - \frac{2}{\pi}\big)$, and then increases for $a\geq \sqrt{\frac{2Q}{\pi}}$. Note that this two-point strategy requires a power cost $P\geq Q\big(1 - \frac{2}{\pi}\big)$ in order to be implemented. This strategy induces a binary interim state $X_1 =  a \cdot \textrm{sign}\big(X_0\big)\in\{-a,a\}$ for which the estimation cost outperforms, in some cases, the best estimation cost obtained via the linear scheme, see \cite[Theorem 2]{W68acis}.

\subsubsection{Dirty Paper Coding (DPC) based scheme for non-causal decoder}\label{sec:DPCStrategy}

In \cite{GS10wcaa}, the authors investigate a vector version of Witsenhausen counterexample in which the decoder is non-causal.

\begin{definition}\label{def:CodeNC}
For $n\in\N^{\star}$, a ``control design'' with non-causal encoder and non-causal decoder is a tuple of stochastic functions $c=(f,g)$ defined by
\begin{align}
f : \mc{X}_0^n  \longrightarrow   \mc{U}_1^n,\quad
g:  \mc{Y}_1^{n}  \longrightarrow \mc{U}_2^n . \label{eq:CodeNC}
\end{align}
We denote by $\mc{C}_{\mathsf{nc}}(n)$ the set of control designs with non-causal encoder and non-causal decoder.
\end{definition}


In \cite[App. D.1-D.7]{GS10wcaa}, the authors investigate a Dirty Paper Coding (DPC) based scheme by using a Gaussian channel input $U_1\sim \mc{N}(0,P)$, $U_1 \perp X_0$ and the auxiliary random variable $W = U_1 + \alpha X_0$. The leads to 
\begin{align}
I(W;Y_1)-I(W;X_0) =&  \frac{1}{2}\log_2\bigg( \frac{P(P+Q+N)}{PQ(1-\alpha)^2 + N(P+\alpha^2 Q)}\bigg),\\
H(U_1+X_0|W,Y_1) =&  \frac{1}{2}\log_2\bigg( (2\pi e) \cdot \frac{NPQ(1-\alpha)^2 }{PQ(1-\alpha)^2 + N(P+\alpha^2 Q)}\bigg).
\end{align}
Since the random variables are jointly Gaussian, the optimization problem writes
\begin{align}
\mathsf{MMSE}_{\mathsf{dpc}}(P) =& \min_{\alpha\in \R,\atop P(P+Q+N)\geq PQ(1-\alpha)^2 + N(P+\alpha^2 Q)}\frac{NPQ(1-\alpha)^2 }{PQ(1-\alpha)^2 + N(P+\alpha^2 Q)}.\label{eq:Optimization_DPC}
\end{align}

\begin{proposition}\label{prop:DPC}
Let $P^{\star}\geq0$, the unique positive root of equation $P^2(P+Q+N) =Q N^2$.\\
$\bullet$ If $P\leq P^{\star}$, the estimation cost for DPC is given by 
\begin{align}
\mathsf{MMSE}_{\mathsf{dpc}}(P) = & \frac{N\big( N\sqrt{Q} - P\sqrt{P+Q+N}\big)^2}{(P+N)^2(P+Q+N)} ,\label{eq:MMSE_DPC}
\end{align}
which is achieved with $\alpha^{\star}= \frac{P(\sqrt{Q}+\sqrt{P+Q+N})}{\sqrt{Q}(P+N)}$.\\
$\bullet$ If $P> P^{\star}$, then $\mathsf{MMSE}_{\mathsf{dpc}}(P) =0$ which is achieved with $\alpha^{\star}=1$.
\end{proposition}
This result is proved in \cite[App. D.1-D.7]{GS10wcaa}, we recall the main proof arguments in App. \ref{sec:ProofPropDPC}. In \cite[App. D.8]{GS10wcaa}, the authors additionally investigate a combination between the linear scheme and the DPC scheme. Given a parameter $-\sqrt{\frac{P}{Q}}\leq \beta \leq \sqrt{\frac{P}{Q}}$, the transmit power $P$ is divided into a linear part $U_{1,1} = -\beta X_0$ and a part $U_{1,2}$ used to implement DPC against the state $(1-\beta)X_0\sim \mc{N}(0,(1-\beta)^2Q)$ with power constraint $\E[U_{1,2}^2]\leq P - \beta^2 Q$. By using the change of variable $\beta = -\rho \sqrt{\frac{P}{Q}}$, we obtain the correlation matrix of \eqref{eq:CorrelationMatrix} and the auxiliary state-dependent channel $Y_1 = \tilde{X}+\tilde{S}+ Z$ where $\tilde{X}+\tilde{S}=X_0+U_1$, the channel state $\tilde{S}$ is defined in \eqref{eq:Stilde}, and  $\tilde{X}\perp (\tilde{S},X_0)$ with $\E[\tilde{X}^2]\leq P(1-\rho^2)$. Therefore, we replace $P$ and $Q$ in \eqref{eq:MMSE_DPC}, respectively by $P(1-\rho^2)$ and $(\sqrt{Q} + \rho \sqrt{P})^2$, and we obtain 
\begin{align}
\mathsf{MMSE}_{\mathsf{lin+dpc}}(P) =& \min_{\rho\in[-1,1]}\frac{N\big(P(1-\rho^2)\sqrt{P+Q + 2 \rho \sqrt{PQ}+N} - N(\sqrt{Q} + \rho \sqrt{P})\big)^2}{(P(1-\rho^2)+N)^2(P+Q + 2 \rho \sqrt{PQ}+N)}.\label{eq:MMSE_lin_DPC}
\end{align}

In the next section, we will see that this estimation cost $\mathsf{MMSE}_{\mathsf{lin+dpc}}(P)$ outperforms all other estimation costs.

\FigMMSEThird{0.1}{0.01}

\subsection{Discussion}

In Figure \ref{fig:OptimalMMSE_0.1_0.01_bis}, we compare the estimation cost proposed in Sec.~\ref{sec:GaussianW1W2} and Sec.~\ref{sec:GaussianW1dDiscreteW2}, with the estimation costs from the literature, for $(Q,N)=(0.1,0.01)$.
\begin{itemize}
\item The blue curve corresponds to the estimation cost of the best linear scheme $\mathsf{MMSE}_{\ell}(P)$ defined in  \eqref{eq:MMSElinear}, see also \cite[Lemma 11]{W68acis}.
\item The green curve depicts the estimation cost of Witsenhausen two point strategy $(P_{\mathsf{two}},\mathsf{MMSE}_{\mathsf{two}})$ defined in \eqref{eq:Power_W} and \eqref{eq:MMSE_W}, see also \cite[Sec.~5]{W68acis}.
\item The red curve presents the estimation cost of the Grover and Sahai's combination of the linear scheme and DPC scheme $\mathsf{MMSE}_{\mathsf{lin+dpc}}(P)$ defined in \eqref{eq:MMSE_lin_DPC}, when the decoder is non-causal, see also \cite[App. D.1-D.8]{GS10wcaa}. 
\end{itemize}

The coordination coding scheme we propose in Sec. \ref{sec:GaussianW1W2} is restricted to Gaussian random variables. The estimation cost $\mathsf{MMSE}_{\mathsf{G}}(P)$ defined in \eqref{eq:SolutionContinuous} consists of the brown line when $Q> 4N $ and $P\in[P_1,P_2]$, and of the estimation cost $\mathsf{MMSE}_{\ell}(P)$ represented by the blue line, otherwise. Note that the function $P\mapsto\mathsf{MMSE}_{\mathsf{G}}(P)$ is the convexification of the linear estimation cost function $P\mapsto \mathsf{MMSE}_{\ell}(P)$.

We reduce the estimation cost by using the auxiliary random variable $W_2=\sign(X_1)$ that encapsulates the sign of the interim state, see Sec.~\ref{sec:GaussianW1dDiscreteW2}. This strategy requires a certain power cost level for the first controller to transmit the sign of $X_1$ to the second controller. The dashed line that is tangent to the orange curve in Fig.~\ref{fig:OptimalMMSE_0.1_0.01_bis}, shows the existence of some weight parameter $  \kappa \in [0,1]$, such that our coordination coding scheme of Sec.~\ref{sec:GaussianW1W2}, outperforms Witsenhausen two point strategy
\begin{align}
\kappa P + (1-\kappa) \mathsf{MMSE}_{\mathsf{coord}}(P) \leq \kappa P_{\mathsf{two}} + (1-\kappa)\mathsf{MMSE}_{\mathsf{two}}.
\end{align}
Note also that the power cost required to implement our coordination coding scheme of Sec.~\ref{sec:GaussianW1W2}, is strictly less than the minimal power cost $Q \big(1 - \frac{2}{\pi}\big)$ needed to implement the Witsenhausen two point strategy. 

When the decoder is non-causal, the combination between the linear scheme and the DPC scheme $\mathsf{MMSE}_{\mathsf{lin+dpc}}(P)$ proposed in \cite{GS10wcaa}, Pareto-dominates all the other solutions.

\section{Conclusion}\label{sec:conclusion}

Our results show that information theoretic methods, in particular coordination coding results, lead to new insights on the Witsenhausen counterexample, and on distributed decision making problems in general. Vice versa, we believe that our observation makes the Witsenhausen counterexample also interesting for other source-channel coding problems. In more detail, we characterize the optimal trade-off between the Witsenhausen power cost and estimation cost, via a single-letter expression with two auxiliary random varibles $(W_1,W_2)$. We show that a convex combination of linear memoryless policies is optimal for the vector-valued Witsenhausen problem with causal decoder, restricted to the space of Gaussian random variables. Since Witsenhausen two-point strategy outperforms the best linear strategy, we investigate a coordination coding with a discrete random variable $W_2 = \sign(X_1)$ and a Gaussian auxiliary random variable $W_1$. For some range of parameters, we show by numerical results that this strategy outperforms both Witsenhausen two-point strategy, and the best linear scheme. In future works, we will consider policies that result in interim states described by more general probability distributions having discrete and continuous parts.


\appendices

\section{Proof of Theorem \ref{theo:MainResult}}
\label{sec:achiev}
\subsection {Preliminaries }
\label{sec:achiev_prelim}

The coordination coding scheme of Sec.~\ref{sec:GaussianW1W2} involves auxiliary random variables that are discrete, e.g. $W_2$, and continuous, e.g. $W_1$. We start with a brief discussion to clarify that the concept of {\it jointly weakly typical} sequences straightforwardly applies to real random vectors with components that are either discrete or continuous, i.e., the measure of discrete random variables is absolutely continuous with the counting measure $\mu$ and the measure of continuous random variables are absolutely continuous with respect to the Lebesgue measure $\lambda$. Since we assume that components in the random vector are either continuous or discrete, the information-theoretic expressions easily follow using the Radon-Nikodym derivative, as comprehensively discussed in \cite[Chapter 2]{Pinsker}.
We briefly recapitulate the definitions refined for our setting for convenience and clarity. 

Let $(X,Y)$ be a pair of real random vectors with $X\in \mathcal{X}\subseteq\R^k$ and $Y\in\mathcal{Y}\subseteq\R^m$
and measure $\nu_{XY}\ll \mu^{\otimes k}\times \lambda^{\otimes m}$, i.e., random vector $X$ is discrete and random vector $Y$ is continuous. Then the Radon-Nikodym derivative of $(X,Y)$ is defined as $f_{XY}(x,y)=\frac{\mathrm{d}\nu_{XY}}{\mathrm{d}
\mu^{\otimes k}\times \lambda^{\otimes m}}(x,y)$, which allows us to define the entropy of $(X,Y)$ as follows
$$
H(X,Y)=-\int\limits_{\mathcal{X}\times\mathcal{Y}
}f_{XY}(x,y)\log f_{XY}(x,y)\,\mathrm{d}(\mu^{\otimes k}\times \lambda^{\otimes m})=
-\sum_{x\in\mathcal{X}}\int\limits_{\mathcal{Y}
}f_{XY}(x,y)\log f_{XY}(x,y)\,\mathrm{d}\lambda^{\otimes m},
$$
where $\mathcal{X}$ and $\mathcal{Y}$ denote the support sets of random vectors $X$ and $Y$. Since we require components to be either discrete or continuous, the Radon-Nikodym derivative reduces to a mixed but 
well separated function with a probability mass function part for the discrete and probability density function part for the continuous variables so that the entropy also reduces to the expected expressions. The expressions for conditional entropy and mutual information then straightforwardly follow considering $X=(X_1,X_2)$ and $Y=(Y_1,Y_2)$ with $X_i\in\R^{k_i}$ and $Y_i\in\R^{m_i}$, $i=1,2$, and $k=k_1+k_2$ and $m=m_1+m_2$. The conditional entropy $H(X_1,Y_1|X_2,Y_2)$ and mutual information 
$I(X_1,Y_1;X_2,Y_2)$ are as usually given by the differences $H(X_1,X_2,Y_1,Y_2)-H(X_2,Y_2)$ and $H(X_1,X_2,Y_1,Y_2)-H(X_1,X_2|Y_1,Y_2)$. 

The standard definition of (jointly) weakly typicality \cite{CT06eoit} applies to random vectors with only discrete or continuous components. For random vectors with both discrete and continuous components, the definition can be just re-stated using the Radon-Nikodym derivative and entropy definition as above. For a two-dimensional random vector with a discrete and continuous component, i.e. $k=m=1$, we denote by $\Aen(X,Y)$ the set of jointly typical sequences given by
$$
\Aen(X,Y)=\left\{(x^n,y^n)\in \R^{2\times n}\,:\,
\Big|-\frac{1}{n}\log f_{XY}^{\otimes n}(x^n,y^n)-H(X,Y)\Big|<\eps, x^n\in\Aen(X), y^n\in\Aen(Y) \right\},
$$
with $f_{XY}^{\otimes n}(x^n,y^n)=\prod_{i=1}^nf_{XY}(x_i,y_i)$. The extension to random vectors with larger dimension, $k+m>2$, is then done as usual. This means, sequences $(x^n,y^n)=(x_1^n,x_2^n,\dots,x_k^n,y_1^n,\dots,y_m^n)\in\mathcal{X}^n\times\mathcal{Y}^n\subseteq \R^{(k+m)\times n}$ need to satisfy the joint typicality condition additionally for any possible subset of the random vector $(X,Y)\in\mathcal{X}\times\mathcal{Y}$. 


Next, it is easy to verify that the joint asymptotic equipartition property straightforwardly extends to random vectors with components that are either discrete or continuous. The reason is that the evaluation is always done component-wise so that the property always follows from the discrete or continuous case. To illustrate this, we 
briefly restate the first two statements of  \cite[Theorem 7.6.1]{CT06eoit}.

{\it Joint AEP.} Let  $(X^n,Y^n)$ be a sequence of length $n$ pair of i.i.d. random vectors with $(X,Y)\in\mathcal{X}\times\mathcal{Y}$
with measure $\nu_{XY}^{\otimes n}
$ and
$\nu_{XY}\ll \mu\times \lambda$. Then
\begin{enumerate}
    \item $\mathbb{P}\{(X^n,Y^n)\in\Aen(X,Y)\}\to 1$ ad $n\to\infty$,
    \item $|\Aen(X,Y)|\leq 2^{n H(X,Y)+\eps}$,
\end{enumerate}
where the notation $\mathbb{P}\{\cdot\}$ stands for the probability of the random event.

The proof is the same as for \cite[Theorem 7.6.1 ]{CT06eoit} since the law of large numbers applies to discrete and continuous random variables. For the second statement considering the continuous part one obviously needs to integrate over the domain instead of taking the sum to compute the volume. The third statement of \cite[Theorem 7.6.1 ]{CT06eoit} can be also reformulated and shown in the same manner. For this, let $(X^n_1,X^n_2,Y^n_1,Y^n_2)$ be a sequence of length $n$ of i.i.d. random vectors  $(X_1,X_2,Y_1,Y_2)\in\mathcal{X}^2\times\mathcal{Y}^2$ with measure $\nu_{X_1X_2Y_1Y_2}^{\otimes n}
$ and $\nu_{X_1X_2Y_1Y_2}\ll \mu^{\otimes2}\times \lambda^{\otimes2}$. Then:
\begin{enumerate}
    \item[3)] If $(\tilde{X}^n_1,\tilde{X}^n_2,\tilde{Y}^n_1,\tilde{Y}^n_2)$ are i.i.d. vectors with measure $\nu_{X_1Y_1}^{\otimes n}\times\nu_{X_2Y_2}^{\otimes n}$, while $\nu_{X_1Y_1}$ and $\nu_{X_2Y_2}$ are the marginals of $\nu_{X_1X_2Y_1Y_2}$ with respect to $(X_1,Y_1)$ and $(X_2,Y_2)$ respectively, then
    $$
    \mathbb{P}\{(\tilde{X}^n_1,\tilde{X}^n_2,\tilde{Y}^n_1,\tilde{Y}^n_2)
    \in\Aen(X_1,X_2,Y_1,Y_2)\}\leq 2^{-n(I(X_1,Y_1;X_2,Y_2)-3\epsilon)},
        $$
        and for sufficiently large $n$,
        $$
    \mathbb{P}\{(\tilde{X}^n_1,\tilde{X}^n_2,\tilde{Y}^n_1,\tilde{Y}^n_2)
    \in\Aen(X_1,X_2,Y_1,Y_2)\}\geq (1-\eps) 2^{-n(I(X_1,Y_1;X_2,Y_2)-3\epsilon)}.
        $$
\end{enumerate}
Likewise, one can straightforwardly verify that all coding lemmas necessary for the following achievablity proof, in particular covering lemma and packing lemma, can be directly restated. Note that results can be always straightforwardly extended like this whenever components in the vectors are separately treated in the proofs. 

In summary, everything straightforwardly extends as expected so that we only need to evaluate the random variables in the expressions with respect to the counting measure (sum) if the random variable is discrete or Lebesgue measure (integral) otherwise. To not unnecessarily over-complicate the expressions in the proof, we will abstain from making this aspect explicit in the following.

\subsection{Achievability Proof}\label{sec:AchievabilityProofCoordination}

The achievability proof uses the block-Markov coding scheme with $B\in \N^{\star}$ blocks each of length $n\in \N^{\star}$ using backward encoding at the encoder and forward decoding at the decoder. The coding scheme follows the empirical coordination scheme with non-causal encoding and causal decoding of \cite{MT17jeco}. Before the regular transmission will be a initialisation phase of length $n'\in \N^{\star}$.
The {\it `error'} analysis is based on the concept of weak typicality 
with an extension that circumvents the need of the Markov Lemma \cite{T78msc} available for strong typicality. A similar approach has been taken in \cite{VOS18hiwp}.

{\it Preliminaries:} Given an arbitrary but fixed $\eps>0$. Further, assume $(X_0^n, X_1^n, U_1^n, U_2^n, W_1^n, W_2^n,  Y_1^n)$ is generated i.i.d. according to the distribution $ \QQ_{X_0X_1U_1U_2W_1W_2Y_1}=\PP_{X_0} \QQ_{U_1W_1W_2|X_0} \PP_{X_1Y_1|X_0U_1} \QQ_{U_2|W_2Y_1}$ of \eqref{eq:TargetDistribution0}, with  $P=\mathbb{E}[U_1^2]$ and $S=\mathbb{E}[(X_1-U_2)^2]$. Then let $\psi^{(n)}:\set{X}_0^n
\times \set{X}_1^n
\times \set{U}_1^n
\times \set{U}_2^n
\times \set{W}_1^n
\times \set{W}_2^n
\times \set{Y}_1^n \to \{0,1\}$ denote an indicator function for sequences of length $n$ with 
\begin{align}
\psi^{(n)}(&x_0^n,x_1^n,  u_1^n,  u_2^n, w_1^n, w_2^n,  y_1^n)=
\begin{cases}
1 & \text{if } |c_P(u_1^n)-P| \geq  \tfrac{1}{2}
\eps\text{ or }|c_S(x_1^n,u_2^n)-S|
\geq \tfrac{1}{12}
\eps\\
&\quad \text{ or } 
(w_1^n,w_2^n, y_1^n)\notin\Aen(W_1,W_2,Y_1),\\
0 & \text{otherwise. }
\end{cases}
\end{align}
Using the weak law of large numbers and the union bound we have
$$
\delta_n=\mathbb{E}[\psi^{(n)}(X_0^n,X_1^n, U_1^n,U_2^n, W_1^n, W_2^n,  Y_1^n)]\overset{n\to \infty}{
\longrightarrow} 0.
$$
Define 
\begin{equation*}
\Sen=\{(x_0^n, w_1^n,w_2^n)\,|\,\eta^{(n)}(x_0^n, w_1^n,w_2^n)\leq \sqrt{\delta_n}\},
\end{equation*}
with 
$
\eta^{(n)}(x_0^n, w_1^n,w_2^n)=\mathbb{E}[
\psi^{(n)}(x_0^n, X_1^n, U_1^n, U_2^n, w_1^n, w_2^n,$
$Y_1^n)
|X_0^n=x_0^n, W_1^n=w_1^n, W_2^n=w_2^n
]$.
Then from the Markov inequality we obtain
\begin{equation*}
\mathbb{P}\{(X_0^n,W_1^n,W_2^n)\notin \set{S}_\eps^{(n)}\}
\leq \frac{
\mathbb{E}[\psi^{(n)}(X_0^n,X_1^n, U_1^n,U_2^n, W_1^n, W_2^n,  Y_1^n)]}{\sqrt{\delta_n}}\leq \sqrt{\delta_n}.
\end{equation*}
We finally define the set 
$$
\Ben = \Aen(X_0,W_1,W_2) \cap  \Sen,
$$
which denotes the set of jointly typical pairs that also satisfy the cost constraints. Note that for $(X_0^n,W_1^n,W_2^n)$ i.i.d.~$\sim \QQ_{X_0W_1W_2}$ we have
 $\mathbb{P}\{(X_0^n,W_1^n,W_2^n)\in \Ben\}\to 1$ as 
 $n\to\infty$. 
 Furthermore, we have the following lemma, which can be similarly shown as in \cite[Lemma 2]{VOS18hiwp}.
\begin{lemma}
\label{LemmaJoint}
Let $X_0^n$ i.i.d.~$\sim \QQ_{X_0}$. For  $M=2^{nR} \geq 2^{n(I(X_0;W_2)+3\eps)}$ codewords $w_2^n(m)$ i.i.d.~$\sim \QQ_{W_2}$, $1\leq m\leq M$ and $L=2^{nR_L} \geq 
2^{n(I(W_1;X_0,W_2)+4\eps)}$ codewords  $w_1^n(\ell,m)$ i.i.d.~$\sim \QQ_{W_1}$, $1\leq \ell\leq L$ and $\eps>0$, we have 
$$
\mathbb{P}\{(X_0^n,W_1^n(1,\ell),W_2^n(m))\notin\Ben \, \forall \, m,\ell\}\to 0\text{ as } n\to\infty.
$$
 \end{lemma}
{\it Proof.}
The proof follows the same arguments as the proof of Lemma 2 in  \cite{VOS18hiwp} with $X$, $U$, and $V$ replaced by $X_0$, $W_1$, and $W_2$ as well as $p_{V|U}$ replaced by $\QQ_{W_1}$ so that (155) changes as follows
\begin{equation*}
\begin{split}
\phantom{xxxx}
\frac{\QQ_{W_1}^{\otimes n}(w_1^n)}{\QQ_{W_1|X_0W_2}^{\otimes n}(w_1^n|x_0^n,w_2^n)}&\geq 
\frac{2^{-n (H(W_1)+ 2\eps)}}{2^{-n(H(W_1|X_0,W_2)- 2\eps)}}\\
&=2^{-n(I(W_1;X_0,W_2)+4\eps}.
\end{split}
\end{equation*}
Thereby, we use the entropy for discrete and continuous parts as defined above so that 
the integration over $w_1^n$ and $w_2^n$ while the integration over the discrete parts using the counting measure become sums.\hfill$\square$

To ensure that the second cost constraint remains bounded even when a coding error happens, the decoder is going to quantize its output.
Since we assume a joint distribution with $\mathbb{E}[(X_1-U_2)^2]= S$, for any $\hat{\delta}>0$ 
there exists a quantization
 $q_{U_2}:\set{U}_2\to \{\hat{u}_{2,k}\}_{k=1}^{N_{U_2}}$
 such that 
 $$
\hat{S}=
\mathbb{E}[(X_1-q_{U_2}(U_2))^2]\leq (1+\hat{\delta})S,
 $$
 in particular such that $\hat{\delta}S<\tfrac{1}{4}\eps$. 

With those preliminaries we are now ready to provide the coding scheme.\\[1ex]


{\it Random codebook:} For rate $R\geq I(X_0;W_2)+ 3\eps$ and rate $R_L\geq I(W_1;W_2;X_0)+4\eps$, generate $2^{nR}$ codewords $w_2^n(m)$ i.i.d. $\sim \QQ_{W_2}$ 
and $2^{n(R+R_L)}$ codewords $w_1^n(m,\ell)$ i.i.d. $\sim \QQ_{W_1}$ with indices $m\in[1:2^{nR}]$ and $\ell\in[1:2^{nR_L}]$. 

{\it Backward encoding at the encoder:} Let $m_b$ and $x_{0,b}^n$ denote the message and processed source sequence of length $n$ of block $b$, $1\leq b \leq B$. Due to non-causal knowledge, the
  encoder performs backward encoding, i.e., the encoder starts with block $b=B$ with initialisation 
  $m_{B+1}=1$ and subsequently encodes the previous blocks.
In block $b$, the encoder takes sequence $x_{0,b}^n$ and message $m_{b+1}$ and looks for $\ell_b$ and $m_b$  such that 
$$
(x_{0,b}^n,w_1^n(m_{b+1},\ell_b), w_2^n(m_b))\in\Ben.
$$
If there are none or more than one pair, then the encoder randomly picks one. Let $w_{1,b}^n=w_1^n(m_{b+1},\ell_b)$  and $w_{2,b}^n=w_2^n(m_b)$ denote the  choice. Next, we generate $u_{1,b}^n\sim \QQ_{U_1|W_1W_2X_0}^{\otimes n}(w_{1,b}^n,w_{2,b}^n,x_0^n)$. 

{\it Forward transmission of the encoder:} In block $b$, $1\leq b \leq B$, if $|c_P(u_{1,b}^n)-P| < \tfrac{1}{4}\eps$
then the encoder transmits $u_{1,b}^n$ synchronously with $x_{0,b}^n$, otherwise the encoder transmits the all zero codeword. The channel distribution $P_{X_1,Y_1|X_0,U_1}^{\otimes n}$ produces channel outputs $x_{1,b}^n$ and $y_{1,b}^n$.

{\it Forward decoding at the decoder:} Let $\tilde{w}_{2,b}^n$ be an abbreviation for  $w_{2,b}^n(\tilde{m}_{b})$ for block $b$, $1\leq b \leq B$, where $\tilde{m}_b$ denotes the index decided on in the previous block $b-1$. Note that message $\tilde{m}_{1}$ will have been obtained from the initialisation phase.  Upon receiving $y_{1,b}^n$, the decoder looks for  $\tilde{\ell}_b$ and $\tilde{m}_{b+1}^n$ such that
$$
(y_{1,b}^n,w_1^n(\tilde{m}_{b+1},\tilde{\ell}_b),\tilde{w}_{2,b})\in\Aen(Y_1,W_1,W_2).
$$
If there are none or more than one pair, then the decoder randomly picks one. 

{\it Forward transmission of the decoder:} In block $b$, $1\leq b \leq B$, the decoder generates  $u_{2,b}^n\sim \QQ_{U_2|W_2Y_1}^{\otimes n}(\tilde{w}_{2,b}^n,y_1^n)$. 
The decoder transmits the quantised sequence $\hat{u}_{2,b}^n$ with elements $\hat{u}_{2,i,b}=q_{U_2}(u_{2,i,b})$ synchronously with $y_{1,b}^n$. 

{\it Sketch for initialisation phase:} Before the first block, message $m_1$ is communicated from the encoder to the decoder using a Gel'fand Pinsker coding scheme, see \cite{GP80cfcw}, treating $X_0^{n'}$ as non-causal channel state knowledge. The auxiliary random variable is picked according to Costa in \cite{Cwodp}, with transmit power $P$ so that the rate $R_{GP}=\frac{1}{2}\log(1+\frac{P}{N})$ is achievable. 
The block length of the initial phase  $n'=\alpha n$ is chosen such that message $m_1$ with rate $R$ can be communicated with an arbitrary small error, i.e., we pick a finite $\alpha>0$ such that $\alpha > R/R_{GP}$. Beside decoding message $m_1$, similarly as in \cite{SCCK05ccas} where the channel state sequence is estimated, the decoder will estimate the evolved state sequence $X_1^{n'}$ using the MMSE estimator 
$$
U_{2,i}=\frac{P+Q}{P+Q+N} \hat{Y}_{1,i},
$$
for $1\leq i \leq n'$. The corresponding mean-squared state estimation error is given by
$$
S'=
\mathbb{E}\left[\frac{1}{n'}\norm{X_1^{n'}-U_2^{n'}}_2^2\right]=\frac{(P+Q)N}{P+Q+N}.
$$
In the following error analysis, the initialisation phase will be denoted as block $b=0$.\\[1ex]


{\it Error analysis per block:}  Let $E^e$ and $E_b^e(m_{b+1})$ denote the events of a failed encoding process and failed encoding in block $b$ given $m_{b+1}$, i.e., $E_b^e(m_{b+1}) = E_b^{e,1}(m_{b+1})\cup E_b^{e,2}$ with $E_b^{e,1}(m_{b+1})=\{ (X_{0,b}^n, W_1^n(m_{b+1},\ell_b), W_2^n(m_b))\notin\Ben
\;\forall (\ell_b, m_b)\}$ and $E_b^{e,2}=\{|c_P(U_{1,b}^n)-P| \geq \frac{1}{4}\eps\}$.
Due to the independence between codewords, the probability of an encoding error in block $b$ given no encoding error in previous blocks does not depend on previous blocks. Accordingly, it is sufficient to analyze the case $m_{b+1}=1$. Thus,
\begin{align}
\mathbb{P}\{E^e_b(M_{b+1})&\,|\, \bigcup\limits_{\beta=b+1}^{B} \bar{E}_\beta^e (M_{\beta+1})\}\\
=&
\mathbb{P}\{E^e_b(1)\}\leq \mathbb{P}\{E^{e,1}_b(1)\}+\mathbb{P}\{E^{e,2}_b\,|\,\bar{E}^{e,1}_b(1)\},
\end{align}
where the bar in $ \bar{E}^e$ denotes the complementary event of $E^e$.
If $R\geq I(X_0;W_2)+3\eps$ and $R_L\geq I(W_1;W_2,X_0)+4\eps$ following Lemma~\ref{LemmaJoint}, 
we have
$
\mathbb{P}\{E^{e,1}_b(1)\}= \mathbb{P}\{(X_{0,b}^n,W_1^n(1,\ell),W_2(m))\notin\Ben \, \forall \, m,\ell\}
\to 0\text{ as } n\to\infty.
$
 Further, we have
$
\mathbb{P}\{E^{e,2}_b|\,\bar{E}^{e,1}_b(1)\}= 
\mathbb{P}\{|c_P(U_{1,b}^n)-P|
\geq \tfrac{1}{4}\eps
\,|\,
(X_{0,b}^n,W_1^n(1,L_b),W_2(m))\in\Ben
\}\leq \sqrt{\delta_n}
\to 0\text{ as } n\to\infty$
due to the law of large numbers.

For the initialisation phase, i.e. block $b=0$,   the encoding and decoding is successful if the message $m_1\in[1:2^{nR}]$ can be successfully send in the initialisation block. This can be done with arbitrary small, but positive probability of error with a sufficiently long block length $n'=\alpha n$ since $\alpha$ has been chosen such that ${nR}<{n'R(1)}=\alpha n R_{GP}$ holds. Thus, we have
$\mathbb{P}\{E^e_0(M_{1})\,|\, \bigcup\limits_{\beta=1}^{B} \bar{E}_\beta^e (M_{\beta+1})\}\to 0$
as well as 
$\mathbb{P}\{E_0^d\}\to 0$
as $n\to\infty$. It follows that $\mathbb{P}\{ E^e\}\to 0$ as $n\to \infty$.

Next, we analyze the decoding error in block $b$, $1\leq b\leq B$. Let $Y_{1,b}^n$ denote the received sequences at the decoder in block $b$. Further, let $E_b^t$ denote the event that sequence $Y_{1,b}^n$  is not jointly typical, i.e., $E_b^t=\{(Y_{1,b}^n,W_1^n(M_{b+1},L_b), \tilde{W}_{2,b}^n)\notin \Aen(Y_1,W_1,W_2)\}$. Then  the decoding error probability $\mathbb{P}\{E^d_b \,|\, \bigcup\limits_{\beta=0}^{b-1} \bar{E}_\beta^d \cup \bar{E}^e \}$
can be upper bounded by
\begin{align}
\mathbb{P}\{E^d_b \,|\, \bigcup\limits_{\beta=0}^{b-1} \bar{E}_\beta^d \cup \bar{E}^e \cup \bar{E}_b^t\}+ 
\mathbb{P}\{E_b^t \,|\, \bigcup\limits_{\beta=0}^{b-1} \bar{E}_\beta^d \cup \bar{E}^e \},
\end{align}
using the union bound. Using the definition of $\Ben$ we obtain the following upper bound for the second term
\begin{align}
\mathbb{P}\{E_b^t \,|\, \bigcup\limits_{\beta=0}^{b-1} \bar{E}_\beta^d \cup \bar{E}^e \}
&
=\mathbb{P}\{Y_{1,b}^n\notin\Aen(Y_1\,|\,W_{1,b}^n, W_{2,b}^n)\,|\,(X_{0,b},W_{1,b}^n, W_{2,b}^n)\in\Ben\}\\
&
\leq \!\!\max\limits_{(x_0^n,w_1^n,w_2^n)\in\Ben} \!\!\eta^{(n)}(x_0^n,w_1^n,w_2^n)\leq \sqrt{\delta_n}\to 0 \text{ as } n\to\infty,
\end{align}
which also ensures that $W_1^n(M_{b+1},L_b)$ will be jointly typical with $Y_{1,b}^n$ and $W_{2,b}^n$. For the correct decoding in block $b$, $1\leq b \leq B$ we have  
\begin{align}
\mathbb{P}\{E^d_b \,|\, \bigcup\limits_{\beta=0}^{b-1} \bar{E}_\beta^d \cup \bar{E}^e \cup \bar{E}_b^t\}&\leq \mathbb{P}\{\exists \tilde{\ell}_b, \tilde{m}_{b+1}\neq M_{b+1}\, :  
W_1^n(\tilde{m}_{b+1},\tilde{\ell}_{b+1}) \in\Aen(W_1|Y_{1,b}^n,W_{2,b}^n)\}\\
&\leq\!\!\!\!\!\!\! \sum_{\substack{\ell_b,\tilde{m}_{b+1}:\\\tilde{m}_{b+1}\neq M_{b+1}}}
\!\!\!\!\!
\max_{(y_1^n,w_2^n)\in\Aen}
\!\!\!\!\!
 \mathbb{P}\{W_1^n(\tilde{m}_{b+1},\tilde{\ell}_{b+1}) \in\Aen(W_1|y_1^n,w_2^n)\}\\
 &\leq 2^{nR} 2^{nR_L} 2^{-n(I(W_1;Y_1,W_2)-3\eps)}=2^{n(R+R_L-I(W_1;Y_1,W_2)+3\eps)},
\end{align}
which goes to 0 as  $n\to\infty$
 if $R+R_L<I(W_1;Y_1,W_2)-3\eps$. It follows that $\mathbb{P}\{ E^d\}\to 0$ as $n\to \infty$.\\[1ex]


{\it Witsenhausen cost analysis:} We first analyze the cost of control.
Let $\psi_b=\psi^{(n)}(X_{0,b}^n,X_{1,b}^n,U_{1,b}^n,U_{2,b}^n,W_{1,b}^n,W_{2,b}^n,Y_{1,b}^n)$ indicate an {\it  error}  in block $b$. From the previous we have  $\mathbb{E}[\psi_b]\to 0$
as $n\to\infty$.
If ${\psi}_b=1$, either for 
the generated input sequence $u_{1,b}^n$ we have $|c_P(u_{1,b}^n)-P|\geq  \tfrac{1}{2}\eps$, or the first cost constraint is satisfied but the second cost constraint or the jointly typicality condition are not satisfied. If the first constraint is not satisfied, then the encoder sets $u_{1,b}^n$ to  the all zero codeword, i.e., bounded error for $\psi_b=1$ so that $\mathbb{E}[|c_P(U_{1,b}^n)-P|]<\eps$ can be shown for $n$ sufficiently large.


Next, for the  estimation error cost we extend the distortion analysis approach by  Wyner \cite{W78trdf}.
Let $\chi_{E,b}$ be an indicator function of the event of an encoding or decoding error in block $b$. 
 From the previous error analysis we have 
 $\mathbb{P}\{\chi_{E,b}\}
 \to 0\text{ as } n\to\infty.$
 Define $\phi_b=(1-\psi_b)(1-\chi_{E,b})$ indicating the  {\it event of desired sequences that satisfy cost and joint typicality constraints} AND {\it no coding error event} in block $b$.  From the previous we have  $\mathbb{E}[{\phi}_b]\to 1$
as $n\to\infty$. In particular, if $\phi_b=1$, then we have $\mathbb{E}[|c_S(X_{1,b}^n,\hat{U}_{2,b}^n)-\hat{S}|]<\tfrac{1}{12}\eps$.  
Therewith, we obtain
\begin{align}
&\mathbb{E}[|c_S(X_{1,b}^n,\hat{U}_{2,b}^n)-\hat{S}]=
\mathbb{E}[\phi_b
|c_S(X_{1,b}^n,\hat{U}_{2,b}^n)-\hat{S}|]+
\mathbb{E}[\bar{\phi}_b
|c_S(X_{1,b}^n,\hat{U}_{2,b}^n)-\hat{S}|]\\
&\leq 
\tfrac{1}{12}\eps 
+ \mathbb{E}[\bar{\phi}_b\hat{S}]
+ \mathbb{E}[\bar{\phi}_bc_S(X_{1,b}^n,\hat{U}_{2,b}^n)].
\end{align}
For $n$ sufficiently large we have
$\mathbb{E}[\bar{\phi}_b\hat{S}]\leq 
\tfrac{1}{12}\eps $ since $\hat{S}<\infty$. The last term can be bounded following Wyner's trick as done in \cite{VOS18hiwp}, which we however need to extend because the internal state $X_1$ instead of source $X_0$ is estimated.

First note that using Cauchy-Schwartz inequality, we have
$\sum_{i=1}^n (a_i+b_i)^2 \leq \sum_{i=1}^n a_i^2+b_i^2 + 2 \sqrt{(\sum_{i=1}^n a_i^2)(\sum_{i=1}^n b_i^2)}$, for any $a_i,b_i\in\mathbb{R}$. Since 
 $X_1=X_0+U_1$, we have
$c_S(X_{1,b}^n,\hat{U}_{2,b}^n)=\frac{1}{n}\sum_{i=1}^n(X_{0i,b}+U_{1i,b}-\hat{U}_{2i,b})^2$.
Associating $U_{1i,b}$ as $a_i$ and $X_{0i,b}-\hat{U}_{2i,b}$ as $b_i$, we obtain the following inequality
\begin{align}
    c_S(X_{1,b}^n,\hat{U}_{2,b}^n)\leq c_P(U_{1,b}^n)+c_S(X_{0,b}^n,\hat{U}_{2,b}^n)
    +2\sqrt{
c_P(U_{1,b}^n)
c_S(X_{0,b}^n,\hat{U}_{2,b}^n)
}.
\end{align}
Further, the encoding ensures that we always have $c_P(U_{1,b}^n)\leq P+
\eps$. Since $\sqrt{\cdot}$ is concave, using Jensen inequality we have 
\begin{align}
\mathbb{E}[\bar{\phi}_bc_S(X_{1,b}^n,\hat{U}_{2,b}^n)]
\leq 
\mathbb{E}\{\bar{\phi}_b 
(P+
\eps+
c_S(X_{0,b}^n,\hat{U}_{2,b}^n))\}
+2\sqrt{
\mathbb{E}[\bar{\phi}_b
(P+
\eps)
c_S(X_{0,b}^n,\hat{U}_{2,b}^n)]
}.
\end{align}
Now, we can argue following Wyner's trick exploiting the discretization of $U_2$ as follows 
\begin{align}
\mathbb{E}[\bar{\phi}_bc_S(X_{0,b}^n,\hat{U}_{2,b}^n)]
=\frac{1}{n}\sum_{i=1}^n
\mathbb{E}[\bar{\phi}_bc_S(X_{0i,b},\hat{U}_{2,i,b})]
\leq \frac{1}{n}\sum_{i=1}^n\mathbb{E}[\bar{\phi}_bD(X_{0,i,b})],
\end{align}
with $D(X_{0,i,b})=\max_{\hat{u}_{2,k}}c_S(X_{0,i,b},\hat{u}_{2,k})$. The random variables $\{D(X_{0,i,b})\}_i$ are  i.i.d. and
 integrable since $c_S(\cdot,\cdot)$ is a squared distance measure and $X_{0,i,b}$ is Gaussian distributed. 
Next, let $\chi_{\{ D(X_{0i,b})>d \}}$ denote an indicator function which is one if $ D(X_{0i,b})>d $. Then we have
\begin{align}
\mathbb{E}[\bar{\phi}_b &
(P+
\eps+
c_S(X_{0,b}^n,\hat{U}_{2,b}^n))]
\leq 
(P+
\eps+d) \,\mathbb{E}[\bar{\phi}_b]
+ 
\mathbb{E}[D(X_{0i,b})
\chi_{\{ D(X_{0i,b})>d \}}
],
\end{align}
as well as
\begin{align}
2\sqrt{\mathbb{E}[\bar{\phi}_b
(P+
\eps)
c_S(X_{0,b}^n,\hat{U}_{2,b}^n)}\leq 
2\sqrt{(P+
\eps) (
\,\mathbb{E}[\bar{\phi}_b]
d + 
\mathbb{E}[D(X_{0i,b})
\chi_{\{ D(X_{0i,b})>d \}}
])}.
\end{align}
Since $D(X_{0i,b})$ is integrable, for any $\eps_d>0$ there must exist a $d_0$ such that 
$\mathbb{E}[D(X_{0i,b}) \chi_{\{ D(X_{0i,b})>d \}}]<\eps_d$ for all $d>d_0$ due to the monotone convergence theorem. Thus for a sufficiently small $\eps_d$ and a sufficiently large $n$ both right hand sides can be upper bounded by $\frac{1}{24}\eps$ so that 
$$
\mathbb{E}[\bar{\phi}_bc_S(X_{1,b}^n,\hat{U}_{2,b}^n)]
\leq \tfrac{1}{12}\eps.
$$
Thus, for the costs of block $b$ we have
\begin{align}
\mathbb{E}[|c_S( X_{1,b}^n,\hat{U}_{2,b}^n)-S|]
&\leq |\hat{S}-S|+ \mathbb{E}[|c_S( X_{1,b}^n,\hat{U}_{2,b}^n)-\hat{S}|]\\
&\leq \tfrac{1}{4}\eps+
\mathbb{E}[|c_S( X_{1,b}^n,\hat{U}_{2,b}^n)-\hat{S}|]
\leq \tfrac{1}{2}\eps,
\end{align}
and $\mathbb{E}\{|{c}(U_{1,b}^n, X_{1,b}^n,U_{2,b}^n)-P-S|\}\leq\tfrac{1}{2}\eps$.

Lastly, we have to include the cost of the initialisation block $b=0$. Since the average transmit power in the initial phase is also set to $P$, we have 
\begin{align}
\mathbb{E}[|c_P(U_{1}^{Bn+n'})-P|]
\leq \frac{\alpha n}{(B+\alpha)n}
\mathbb{E}[|c_P(U_{1,0}^{n'})-P|]+
\frac{n}{(B+\alpha)n}
\sum_{b=1}^B
\mathbb{E}[|c_P(U_{1,b}^{n})-P|]
\leq 
\eps.
\end{align}
For the estimation error, the initial phase results in a larger but bounded error average error $S'<\infty$. The impact however can be made arbitrary small  with a sufficiently large number of blocks $B$  as follows
\begin{align}
\mathbb{E}[|c_S(X_1^{Bn+n'}, U_{2}^{Bn+n'})-S|]
\leq\frac{\alpha n}{(B+\alpha)n}\mathbb{E}[|c_S(X_1^{n'},U_{2}^{n'})
-S|]
+\frac{n}{(B+\alpha)n} 
\sum_{b=1}^B
\mathbb{E}[|c_S(X_{1,b}^{n}, U_{2,b}^{n})-S|]
\leq \eps,
\end{align}
for $n$ and $B$ sufficiently large. 

Lastly, the existence of a coordination scheme follows from the extension of the random coding argument as in the proof of \cite[Lemma 2.2]{BB11pls}.


{\it Closedness:} The previous holds if the rate constraint holds with strict inequality. For equality, we can argue as in \cite[App.C]{MT17jeco}, i.e., since $N<\infty$, we can always find an approximation of the random variables $W_1, W_2, U_1$ and $U_2$ that result in an arbitrary small increase of the costs, but satisfy the rate constraint with strict inequality.\hfill $\square$

\subsection{Converse proof}\label{sec:ConverseProofCoordination}

The converse proof follows the same arguments as in \cite[Sec.V-B]{MT15ctec}. We consider a control design $c\in\mc{C}_{\mathsf{d}}(n)$ of block-length $n\in\N^{\star}$ such that $\gamma^n_{\mathsf{p}}(c)<+\infty$ and $\gamma^n_{\mathsf{s}}(c)<+\infty$. According to Csisz\'{a}r sum identity, see \cite[pp.25]{EK11nit}, we have
\begin{align}
0=&  \sum_{t=1}^n I(   X^n_{0,t+1} ; Y_{1,t}       | Y_1^{t-1}) - \sum_{t=1}^n I( Y_1^{t-1}   ; X_{0,t}| X^n_{0,t+1} ) \label{eq:Conv1} \\
=&  \sum_{t=1}^n I(   X^n_{0,t+1} ; Y_{1,t}       | Y_1^{t-1}) - \sum_{t=1}^n I(X^n_{0,t+1} , Y_1^{t-1}   ; X_{0,t}) \label{eq:Conv2} \\
=&  \sum_{t=1}^n I(W_{1,t}; Y_{1,t} | W_{2,t} ) - \sum_{t=1}^n I(W_{1,t}, W_{2,t} ; X_{0,t}),\label{eq:Conv3} \\
=&  n\cdot\Big(  I(W_{1,T}; Y_{1,T} | W_{2,T}, T ) - I(W_{1,T}, W_{2,T} ; X_{0,T}|T)\Big)\label{eq:Conv4} \\
\leq &  n\cdot\Big(  I(W_{1,T},T; Y_{1,T} | W_{2,T} ) - I(W_{1,T},T, W_{2,T} ; X_{0,T})\Big)\label{eq:Conv5} \\
\leq &  n\cdot\Big(  I(W_{1}; Y_{1} | W_{2} ) - I(W_{1}, W_{2} ; X_{0})\Big)\label{eq:Conv6} 
\end{align}
Equation \eqref{eq:Conv2} comes from the i.i.d. property of the state.\\ 
Equation \eqref{eq:Conv3} comes from the identification of the auxiliary random variables $W_{1,t} = X^n_{0,t+1} $ and $W_{2,t} = Y_1^{t-1}$, for $t\in\{1,\ldots,n\}$.\\
Equation \eqref{eq:Conv4} comes from the introduction of 
the uniform random variable $T \in \{1,\ldots,n\}$ and the auxiliary random variables $X_{0,T}$,  $W_{1,T}$, $W_{2,T}$ and $Y_{1,T}$, where $Y_{1,T}$ is distributed according to 
\begin{align}
\mathbb{P}\big\{Y_{1,T} = y_1\big\}  = \frac1n\sum_{t=1}^n \mathbb{P}\big\{Y_{1,t} = y_1\big\},\qquad \forall y_1\in \mc{Y}_1.\label{eq:AuxRVT}
\end{align}
Equation \eqref{eq:Conv5} comes from the independence between the random variables $T$ and $X_{0,T}$.\\
Equation \eqref{eq:Conv6} comes from the introduction of the auxiliary random variables $X_0=X_{0,T}$, $Y_1=Y_{1,T}$, $W_1=(W_{1,T},T)$, $W_2 =W_{2,T}$.

We show that the auxiliary random variables $W_{1,t} = X^n_{0,t+1} $ and $W_{2,t} = Y_1^{t-1}$ satisfy the following Markov chains, for $t\in\{1,\ldots,n\}$. 
\begin{align}
&(X_{1,t},Y_{1,t})  -\!\!\!\!\minuso\!\!\!\!- (X_{0,t} ,U_{1,t}) -\!\!\!\!\minuso\!\!\!\!-  (W_{1,t}, W_{2,t}) ,\label{eq:MarkovChain1}\\ 
&U_{2,t} -\!\!\!\!\minuso\!\!\!\!- (Y_{1,t} , W_{2,t} ) -\!\!\!\!\minuso\!\!\!\!-  (X_{0,t},X_{1,t},U_{1,t}, W_{1,t} ) . \label{eq:MarkovChain2} 
\end{align}
Equation  \eqref{eq:MarkovChain1} comes from the memoryless property of the channel $\QQ_{X_1Y_1|X_0U_1}$.\\
Equation  \eqref{eq:MarkovChain2} comes from the causal decoding: for all $t\in\{1,\ldots,n\}$, the output of the decoder $U_{2,t}$ depends on the symbols $(X_{0,t},X_{1,t},U_{1,t}, X^n_{0,t+1}  ) $ only through the past and current channel outputs $( Y_{1,t}, Y_1^{t-1} )$. 

This implies that the auxiliary random variables $X_0=X_{0,T}$, $U_1=U_{1,T}$,  $X_1=X_{1,T}$, $Y_1=Y_{1,T}$, $W_1=(W_{1,T},T)$, $W_2 =W_{2,T}$, $U_2=U_{2,T}$ satisfy the  following Markov chains. 
\begin{align}
&(X_{1},Y_{1})  -\!\!\!\!\minuso\!\!\!\!- (X_{0} ,U_{1}) -\!\!\!\!\minuso\!\!\!\!-  (W_{1}, W_{2}) ,\label{eq:MarkovChain11}\\ 
&U_{2} -\!\!\!\!\minuso\!\!\!\!- (Y_{1} , W_{2} ) -\!\!\!\!\minuso\!\!\!\!-  (X_{0},X_{1},U_{1}, W_{1} ) . \label{eq:MarkovChain22} 
\end{align}
Therefore, the distribution of the auxiliary random variables decomposes as in \eqref{eq:TargetDistribution0}. 

We reformulate the $n$-stage costs by using the auxiliary random variables $U_1=U_{1,T}$,  $X_1=X_{1,T}$, $U_2=U_{2,T}$ and \eqref{eq:AuxRVT}.
\begin{align}
\gamma^n_{\mathsf{p}}(c) =& \E\Bigg[\frac{1}{n}\sum_{t=1}^n {U_{1,t}}^2\Bigg] =   \E\Big[ {U_{1,T}}^2\Big] =   \E\Big[ U_{1}^2\Big]\label{eq:GammaN_p4}.
\end{align}
By using a similar argument, we show that
\begin{align}
\gamma^n_{\mathsf{s}}(c) =& \E\Big[ \big(X_1-U_2\big)^2\Big]\label{eq:GammaN_s4}.
\end{align}

In conclusion, if the pair of costs $(P,S)\in\R^2$ is achievable, then for all $\varepsilon>0$, there exists $\bar{n}\in\N^{\star}$ such that for all $n\geq \bar{n}$, there exists a control design $c\in\mc{C}_{\mathsf{d}}(n)$ such that \eqref{eq:TargetDistribution0} and \eqref{eq:TheoIC}  are satisfied and 
\begin{align}
\Big| P - \E\Big[ U_{1}^2\Big] \Big| + \Big| S - \E\Big[ \big(X_1-U_2\big)^2\Big] \Big| \leq \varepsilon.\label{eq:EqualConstraint0}
\end{align}
The equation \eqref{eq:EqualConstraint0} is valid for all $\varepsilon>0$. This concludes the converse proof of Theorem \ref{theo:MainResult}.

\section{Proof of Lemma \ref{lemma:LinearScheme}}\label{sec:ProofLemmaLinear}

We consider the linear strategy
\begin{align}
U_1 = a \cdot X_0 +b,
\end{align}
with parameters $(a,b)\in \R^2$. It induces an interim state random variable
\begin{align}
X_1 = (1+a) \cdot X_0 +b,
\end{align}
which is distributed according to $\mathcal{N}\big(b,(1+a)^2Q\big)$. Since $Y_1 = X_1 + Z_1$ with $Z_1\sim\mc{N}(0,N)$, the conditional probability density function of $X_1$ given a realization $Y_1=y_1$ is given by
\begin{align}
f(x_1|y_1) = \frac{1}{\sqrt{\frac{(1+a)^2QN}{(1+a)^2Q+N}}}\phi \Bigg(\frac{x_1 - y_1\cdot \frac{(1+a)^2Q}{(1+a)^2Q+N} - \frac{b N }{(1+a)^2Q + N}}{\sqrt{\frac{(1+a)^2QN}{(1+a)^2Q + N}}}\Bigg), \quad \forall x_1\in \R,\label{eq:ConditionalDistribution}
\end{align}
where $\phi(x)=  \frac{1}{\sqrt{2 \pi}}\exp\big(-\frac{x^2}{2}\big)$. Therefore, we have 
\begin{align}
&\E\Big[{U_1}^2\Big] = a^2 Q + b^2,\label{eq:PowerCostLemma}\\
&\E\Big[\big(X_1 - \E[X_1|Y_1]\big)^2\Big|Y_1\Big] = \frac{(1+a)^2QN}{(1+a)^2Q+N}.\label{eq:EstimationCostLemma}
\end{align}
The estimation cost in \eqref{eq:EstimationCostLemma} does not depend on the parameter $b\in \R$. The function $a\mapsto \frac{(1+a)^2QN}{(1+a)^2Q+N}$ is strictly decreasing over the interval $]-\infty,-1[$, reaches zero in $a=-1$, and is strictly increasing over $] -1,+\infty[$.

Suppose that $P> Q$, then we select $a=-1$ and $b\in \R$ such that 
\begin{align}
 Q + b^2 = P \quad \Longleftrightarrow \quad b \in \Big\{ \sqrt{P-Q}, - \sqrt{P-Q}\Big\},
\end{align}
which induces an estimation cost in \eqref{eq:EstimationCostLemma} equal to zero.

Suppose that $P \leq Q$ and assume that $b\neq0$. For all parameters $a\in R$ and $0<b^2$ such that $a^2 Q + b^2=P$, there exists other parameters $a' = -\sqrt{\frac{P}{Q}}<a$ and $b'=0$ such that ${a'}^2 Q + {b'}^2=a^2 Q + b^2=P$ and the estimation cost in \eqref{eq:EstimationCostLemma} is reduced. Therefore, at the optimum we must have $b=0$ and $a = -\sqrt{\frac{P}{Q}}$. Therefore, the best linear scheme is defined by 
\begin{align}
U_{1,\ell}(P)=& 
\begin{cases}
-\sqrt{\frac{P}{Q}} X_0&\text{ if } P\in[0,Q],\\
-X_0+\sqrt{P-Q} \qquad\qquad &\text{ otherwise.}
\end{cases}\label{eq:MMSElinear00}
\end{align}
It induces an estimation cost 
\begin{align}
\inf_{(a,b)\in R^2,\atop a^2 Q + b^2 = P}\E\Big[(X_1 - \E[X_1|Y_1])^2)\Big|Y_1\Big] 
=&\begin{cases}
\frac{ \big(\sqrt{Q}-  \sqrt{P} \big)^2 \cdot N}{ \big(\sqrt{Q}-  \sqrt{P} \big)^2 + N}\quad &\text{ if } P\in[0,Q],\\
0 \qquad\qquad &\text{ otherwise,}
\end{cases}\label{eq:MMSElinear0}
\end{align}
that corresponds to the definition of $\mathsf{MMSE}_{\ell}(P)$.

\section{Proof of Theorem \ref{theo:CharacterizationContinuousRV}}\label{sec:ProofTheoGaussian}

Throughout the proof, we assume that $P\leq Q$.

\subsection{Lower bound}\label{sec:ConverseProof}

The Markov chain $Y_1 -\!\!\!\!\minuso\!\!\!\!- (X_0,U_1)  -\!\!\!\!\minuso\!\!\!\!- (W_1,W_2)$ implies
\begin{align}
I(W_1;Y_1|W_2) - I(W_1,W_2;X_0) 
\leq& I(W_1;Y_1|W_2,X_0) - I(W_2;X_0)\\
\leq& I(U_1;Y_1|W_2,X_0) - I(W_2;X_0).
\end{align}
Therefore
\begin{align}
\mathsf{MMSE}_{\mathsf{G}}(P) 
\geq& \min_{\QQ_{U_1W_2|X_0}  \in \Q_1(P)} \E_{\QQ}\Big[\Big(X_{1} - \E\big[X_1|W_2,Y_1\big]\Big)^2\Big| W_2,Y_1\Big], \label{eq:MMSEweak}
\end{align}
where 
\begin{align}
\Q_1(P) = \bigg\{  &\QQ_{U_1W_2|X_0} \;\text{ s.t. }\; P = E_{\QQ}\big[U_{1}^2\big],\nonumber \\
& I(U_1;Y_1|W_2,X_0) - I(W_2;X_0)\geq0,  \nonumber \\
& (X_0,U_1,W_2,X_1,Y_1, U_2)\;  \text{ are Gaussian  }\bigg\},
\end{align}

According to \eqref{eq:EstimationCostLemma}, the estimation cost does not depend on the mean vector of the Gaussian random variables. Without loss of generality, we consider that the Gaussian random variables $(X_0,W_2,U_1)\sim  \mc{N}(0,K)$ optimal for \eqref{eq:MMSEweak}, are centered with covariance matrix
\begin{align}
K = 
\begin{pmatrix}
Q & \rho_1\sqrt{QV} & \rho_2\sqrt{QP} \\
\rho_1\sqrt{QV} &V & \rho_3\sqrt{VP} \\
 \rho_2\sqrt{QP} &  \rho_3\sqrt{VP} & P\\
\end{pmatrix},\label{eq:CovarianceMatrix}
\end{align}
where the correlation coefficients $(\rho_1,\rho_2,\rho_3)\in[-1,1]^3$ are such that $\det (K) = Q V  P \cdot \big(1- {\rho_1}^2- {\rho_2}^2- {\rho_3}^2 + 2 \rho_1 \rho_2\rho_3\big)\geq0$, i.e. $K$ is semi-definite positive.

\begin{lemma}\label{lemma:EntropyComputationRho}
Assume that $(X_0,W_2,U_1)\sim  \mc{N}(0,K)$, then
\begin{align}
I(U_1;Y|X_0,W_2)  - I(X_0;W_2)  
&=  \frac{1}{2 }\log_2 \bigg( \frac{ P}{ N}  \cdot (1- {\rho_1}^2 - {\rho_2}^2 - {\rho_3}^2 + 2 \rho_1\rho_2\rho_3) + (1- {\rho_1}^2)\bigg),\label{eq:IC}\\
\E_{\QQ}\Big[\Big(X_{1} - \E\big[X_1|W_2,Y_1\big]\Big)^2\Big| W_2,Y_1\Big] 
&=  \frac{N  \Big(Q (1-{\rho_1}^2 )+  P(1-{\rho_3}^2)+ 2\sqrt{QP}  (\rho_2 -  \rho_1 \rho_3)\Big)}{N + \Big( Q (1-{\rho_1}^2) +  P (1-{\rho_3}^2) +  2\sqrt{QP}  (\rho_2 -  \rho_1 \rho_3) \Big)}.\label{eq:MMSE}
\end{align}
\end{lemma}
The proof of Lemma \ref{lemma:EntropyComputationRho} is stated in Sec. \ref{sec:ProofLemmaEntropyComputation}.
Note that \eqref{eq:IC} and \eqref{eq:MMSE} do not depend on the variance parameter $V$ of the auxiliary random variable $W_2$. Moreover, if \eqref{eq:IC} is positive, then the matrix $K$ is semi-definite positive.

By using Lemma \ref{lemma:EntropyComputationRho}, we reformulate \eqref{eq:MMSEweak} and since the function $x\to \frac{N\cdot x}{N+ x}$ is strictly increasing for all $x\geq0$, the optimal parameters $(\rho_1^{\star}, \rho_2^{\star}, \rho_3^{\star})\in[-1,1]^3$ minimize
\begin{align}
Q (1-{\rho_1}^2 )+  P(1-{\rho_3}^2)+ 2\sqrt{QP}  (\rho_2 -  \rho_1 \rho_3) ,\label{eq:ConverseMinimizationCriteriaPb22}
\end{align}
under the constraint
\begin{align}
&\frac{ P}{ N}  \cdot (1- {\rho_1}^2 - {\rho_2}^2 - {\rho_3}^2 + 2 \rho_1\rho_2\rho_3) - {\rho_1}^2  \geq 0\label{eq:ConverseMinimizationPb2}\\
\Longleftrightarrow & (1- {\rho_1}^2)\cdot (1  - {\rho_3}^2) -  \frac{N}{P}  \cdot {\rho_1}^2 \geq   (\rho_2 -  \rho_1\rho_3)^2   ,
\end{align}
which yields the optimal parameter
\begin{align}
\rho_2^{\star}  =&   \rho_1\rho_3  - \sqrt{(1- {\rho_1}^2)\cdot (1  - {\rho_3}^2) -  \frac{N}{P}  \cdot {\rho_1}^2}.\label{eq:Rho2}
\end{align}

\begin{lemma}\label{lemma:OptimalRho1Rho3}
If $Q> 4N$ and $P\in[P_1,P_2]$, then 
\begin{align}
{\rho_1^{\star}}= \sqrt{\frac{P Q - (P+N)^2}{Q(P+N)}},\qquad
\rho_2^{\star} = - \frac{P+N}{\sqrt{PQ}},\qquad
{\rho_3^{\star}}^2=0.
\end{align}
If $Q\leq4N$ \emph{or} if $Q> 4N$ and $P\in[0,P_1]\cup[P_2,Q]$, then 
\begin{align}
{\rho_1^{\star}}=0,\qquad
{\rho_2^{\star}}=-1,\qquad
{\rho_3^{\star}}=0.
\end{align}
\end{lemma}
The proof of Lemma \ref{lemma:OptimalRho1Rho3} is stated in App. \ref{sec:ProofLemmaOptimalRho}.
We obtain the lower bound by replacing the optimal parameters $(\rho_1^{\star}, \rho_2^{\star}, \rho_3^{\star})\in[-1,1]^3$ in \eqref{eq:MMSE}. For all $P\leq Q$, we have
\begin{equation}
\mathsf{MMSE}_{\mathsf{G}}(P) \geq
\begin{cases}
\frac{N\cdot(Q-N-P)}{Q} & \text{if } Q> 4N \text{ and } P\in[P_1,P_2],\\
\frac{ \big(\sqrt{Q}-  \sqrt{P} \big)^2 \cdot N}{ \big(\sqrt{Q}-  \sqrt{P} \big)^2 + N}& \text{otherwise. }
\end{cases}\label{eq:SolutionContinuousConverse}
\end{equation}


\subsection{Upper bound}\label{sec:AchievabilityProof}

\subsubsection{Linear Scheme}

According to the Lemma \ref{lemma:LinearScheme} for $P\leq Q$, the optimal linear scheme is given by
\begin{align}
U_1 = - \sqrt{\frac{P}{Q}} \cdot X_0.
\end{align}
Therefore, $\mathsf{MMSE}_{\mathsf{G}}(P) \leq \mathsf{MMSE}_{\ell}(P) $, for all $P\leq Q$.

\subsubsection{Case where $Q> 4N$ and $P\in[P_1,P_2]$}

The upper bound of Theorem \ref{theo:CharacterizationContinuousRV} can be obtained by using a time sharing strategy between the two linear schemes with parameters $P_1$ and $P_2$ defined by 
\begin{align}
P_1 =& \frac{1}{2}\cdot \Big(Q - 2N - \sqrt{Q\cdot(Q-4N)}\Big),\\
P_2 =& \frac{1}{2}\cdot \Big(Q - 2N +\sqrt{Q\cdot(Q-4N)}\Big).
\end{align}

We show that we obtain the same result by replacing the coefficients of Lemma \ref{lemma:OptimalRho1Rho3}
\begin{align}
{\rho_1}= \sqrt{\frac{P Q - (P+N)^2}{Q(P+N)}},\qquad
\rho_2 = - \frac{P+N}{\sqrt{PQ}},\qquad
{\rho_3}=0,\label{eq:OptimalRhoAchieva}
\end{align}
in the covariance matrix $K$ of the random variables $(X_0,W_2,U_1)\sim  \mc{N}(0,K)$, in \eqref{eq:CovarianceMatrix}, and by selecting accurately Costa's auxiliary random variable $W_1$ for a Dirty Paper Coding (DPC), see \cite{Cwodp}.

Since the random variable $W_2$ is correlated with the state $X_0$, we have
\begin{align}
W_2 = & \rho_1 \sqrt{\frac{V}{Q}} X_0 + Z_0, \quad Z_0 \sim \mathcal{N}\big(0, V(1 - {\rho_1}^2)\big),\quad  Z_0 \perp X_0,\label{eq:AchievabilityW2s}
\end{align}
and then
\begin{align}
I(X_0;W_2) = \frac12\log_2 \bigg(\frac{1}{1-{\rho_1}^2}\bigg).\label{eq:ICrho1}
\end{align}
The channel input $U_1$ reformulates as
\begin{align}
U_1 =& \frac{\rho_2 - \rho_1\rho_3}{1 - {\rho_1}^2}\sqrt{\frac{P}{Q}} X_0 + \frac{\rho_3 - \rho_1\rho_2}{1 - {\rho_1}^2}\sqrt{\frac{P}{V}} W_2 + U_0,\label{eq:AchievabilityU1s}\\
\text{with }&\quad U_0 \sim \mathcal{N}\bigg(0,P\cdot \frac{1 - {\rho_1}^2 - {\rho_2}^2 - {\rho_3}^2 + 2 {\rho_1}{\rho_2}{\rho_3}}{1 - {\rho_1}^2}\bigg),\quad  U_0 \perp (X_0,W_2).\label{eq:U0}
\end{align}

In order to evaluated the information constraint, we now state two lemmas.
\begin{lemma}\label{lemma:MultiplicativeEntropy}
Assume that $(X_0,W_2,U_1)\sim  \mc{N}(0,K)$ and let $\tilde{W}_2 = \beta W_2$ with $\beta\in \R$. Then
\begin{align}
H(X_0,\tilde{W}_2 ) =&H(X_0,{W}_2 ) + \log_2|\beta|,\\
H(X_0,\tilde{W}_2,U_1 ) =& H(X_0,{W}_2,U_1 ) + \log_2|\beta|.
\end{align}
\end{lemma}
The proof of Lemma \ref{lemma:MultiplicativeEntropy} is stated in Sec. \ref{sec:lemma:MultiplicativeEntropy}.

\begin{lemma}\label{lemma:StateDependent}
Consider the state-dependent channel 
\begin{align}
\tilde{Y_1} = \tilde{X}_0 + \tilde{W}_2 +  \tilde{U}_0 + \tilde{Z},
\end{align}
with Gaussian channel-state parameters $(\tilde{X}_0,\tilde{W}_2)\sim  \mc{N}(0,\tilde{K})$ and
\begin{align}
\tilde{K} = 
\begin{pmatrix}
q    &  \mu\sqrt{qv}\\
\mu\sqrt{qv} &v
\end{pmatrix},
\end{align} 
with $q\geq0$, $v\geq0$, $\mu\in[-1,1]$, and with Gaussian noise $\tilde{Z}\sim\mc{N}(0,N)$ such that $\tilde{Z}\perp( \tilde{X}_0, \tilde{W}_2, \tilde{U}_0)$. We assume that the channel input is also Gaussian $\tilde{U}_0\sim\mc{N}(0,P_0)$, $P_0\geq0$, with $\tilde{U}_0\perp(\tilde{X}_0,\tilde{W}_2)$, and we introduce Costa's auxiliary random variable, see \cite{Cwodp}, 
\begin{align}
 \tilde{W}_1 = \tilde{U}_0 + \alpha \tilde{X}_0, \qquad \alpha \in \R.
\end{align}
Then, 
\begin{align}
I( \tilde{W}_1;\tilde{Y_1},\tilde{W}_2) - I( \tilde{W}_1;\tilde{X}_0,\tilde{W}_2) = \frac12\log_2 \Bigg(\frac{P_0\big(q(1-\mu^2)+P_0+N\big)}{P_0N+q(1-\mu^2)\big((1-\alpha)^2 P_0+ \alpha^2 N\big)}\Bigg).\label{eq:ICtilde}
\end{align}
\end{lemma}

The proof of Lemma \ref{lemma:StateDependent} is stated in Sec. \ref{sec:lemma:StateDependent}. We select $P_0= P\cdot \frac{1 - {\rho_1}^2 - {\rho_2}^2 - {\rho_3}^2 + 2 {\rho_1}{\rho_2}{\rho_3}}{1 - {\rho_1}^2}$, $\alpha = \frac{P_0}{P_0+N}$ and we identify $\tilde{U}_0=U_0$ given in \eqref{eq:U0}, $\tilde{Z}= Z$, and 
\begin{align}
\tilde{X}_0 = \Bigg(1+ \frac{\rho_2 - \rho_1\rho_3}{1 - {\rho_1}^2}\sqrt{\frac{P}{Q}} \Bigg) \cdot X_0 ,\qquad 
\tilde{W}_2 =  \frac{\rho_3 - \rho_1\rho_2}{1 - {\rho_1}^2}\sqrt{\frac{P}{V}} \cdot  W_2.
\end{align}
These choices of parameters imply that $\tilde{Y_1} = Y_1$, $\mu=\rho_1$ and
\begin{align}
q= \Bigg(\sqrt{Q}+ \sqrt{P}\frac{\rho_2 - \rho_1\rho_3}{1 - {\rho_1}^2} \Bigg)^2,\qquad 
v = P \Bigg(\frac{\rho_3 - \rho_1\rho_2}{1 - {\rho_1}^2}\Bigg)^2.
\end{align}
We define the auxiliary random variable $W_1$ by
\begin{align}
W_1 = \tilde{W}_1 = \tilde{U}_0 + \alpha \tilde{X}_0 = U_0 + \alpha \Bigg(1+ \frac{\rho_2 - \rho_1\rho_3}{1 - {\rho_1}^2}\sqrt{\frac{P}{Q}} \Bigg) \cdot X_0 .
\end{align}

According to Lemmas \ref{lemma:StateDependent} and \ref{lemma:MultiplicativeEntropy}, and since $\alpha = \frac{P_0}{P_0+N}$ implies $(1-\alpha)^2 P_0+ \alpha^2 N = \frac{P_0 N}{P_0+N}$, we have 
\begin{align}
&I(W_1;{Y_1},{W}_2) - I(W_1;{X}_0,{W}_2)\nonumber\\
=&
H({Y_1},{W}_2) - H(W_1,{Y_1},{W}_2) +H(W_1|X_0,{W}_2)\\
=&
H({Y_1},\tilde{W}_2) - H(W_1,{Y_1},\tilde{W}_2) +H(W_1|\tilde{X}_0,\tilde{W}_2)\\
=& \frac12\log_2\Bigg(\frac{P_0\big(q(1-\mu^2)+P_0+N\big)}{P_0N+q(1-\mu^2)\big((1-\alpha)^2 P_0+ \alpha^2 N\big)}\Bigg)\\
=&\frac12\log_2\Bigg(\frac{P_0\big(q(1-\mu^2)+P_0+N\big)}{P_0N+q(1-\mu^2)\frac{P_0 N}{P_0+N}}\Bigg)\\
=&\frac12\log_2\Bigg(\frac{P_0\big(q(1-\mu^2)+P_0+N\big)}{\frac{P_0 N}{P_0+N}\big(P_0+N+q(1-\mu^2)\big)}\Bigg)\\
=&\frac12\log_2\Bigg(1+ \frac{P_0}{N}\Bigg)\\
=&\frac12\log_2\Bigg(1+ \frac{P}{N}\cdot\frac{1 - {\rho_1}^2 - {\rho_2}^2 - {\rho_3}^2 + 2 {\rho_1}{\rho_2}{\rho_3}}{1 - {\rho_1}^2}\Bigg).\label{eq:ICrho2}
\end{align}

Now, we replace in \eqref{eq:ICrho1} and \eqref{eq:ICrho2}, the coefficients of Lemma \ref{lemma:OptimalRho1Rho3}
\begin{align}
{\rho_1}= \sqrt{\frac{P Q - (P+N)^2}{Q(P+N)}},\qquad
\rho_2 = - \frac{P+N}{\sqrt{PQ}},\qquad
{\rho_3}=0.\label{eq:OptimalRhoAchievaBB}
\end{align}
We obtain
\begin{align}
&I(W_1;W_2,Y_1) - I(W_1;X_0,W_2)  
=  \frac{1}{2 }\log_2 \bigg(\frac{Q(P+N)}{QN + (P+N)^2}\bigg)
= I(X_0;W_2) = I(U_1;Y_1|X_0,W_2).\label{eq:INC}
\end{align}
Equation \eqref{eq:INC} ensures that the combination of the lossy source coding of $X_0$ via $W_2$, with Costa's coding, see \cite{Cwodp}, for state-dependent channel $Y_1$, is achievable. According to Lemma \ref{lemma:EntropyComputationRho}, we have
\begin{align}
\E_{\QQ}\Big[\Big(X_{1} - \E\big[X_1|W_2,Y_1\big]\Big)^2\Big| W_2,Y_1\Big]  = \frac{N\cdot(Q-N-P)}{Q}.\end{align}

\subsection{Proof of Lemma \ref{lemma:EntropyComputationRho}}\label{sec:ProofLemmaEntropyComputation}

We consider  $(X_0,W_2,U_1)\sim \mathcal{N}(0,K)$ with $K$ defined in \eqref{eq:CovarianceMatrix}, which together with \eqref{eq:Gaussian3}, induces the  Gaussian random variables $(X_0,W_2,Y_1)$ whose entropy is
\begin{align}
h(X_0,W_2,Y) =\frac{1}{2 }\log_2 \bigg((2\pi e)^3\cdot Q V  \Big( P (1- {\rho_1}^2 - {\rho_2}^2 - {\rho_3}^2 + 2 \rho_1\rho_2\rho_3) + N (1- {\rho_1}^2)  \Big) \bigg).
\end{align}
Therefore we have
\begin{align}
I(U_1;Y|X_0,W_2)  - I(X_0;W_2)
=& h(X_0,W_2,Y)  - h(Y|U_1,X_0,W_2) - h(X_0) - h(W_2)\\
=&  \frac{1}{2 }\log_2 \bigg( \frac{ P}{ N}  \cdot (1- {\rho_1}^2 - {\rho_2}^2 - {\rho_3}^2 + 2 \rho_1\rho_2\rho_3) + (1- {\rho_1}^2)\bigg). 
\end{align}
According to \eqref{eq:Gaussian1} and \eqref{eq:Gaussian3} the entropy of $(X_1,W_2,Y_1)$ writes
\begin{align}
h(X_1,W_2,Y) = \frac{1}{2 }\log_2 \bigg((2\pi e)^3\cdot   V  N \Big(Q  (1-{\rho_1}^2 )+  P  (1-{\rho_3}^2)+ 2\sqrt{QP}  (\rho_2 -  \rho_1 \rho_3)\Big)\bigg),
\end{align}
and hence 
\begin{align}
\E\Big[\big(X_1 - \E[X_1|W_2,Y_1]\big)^2\Big|W_2,Y_1\Big] = \frac{N \Big(Q  (1-{\rho_1}^2 )+  P (1-{\rho_3}^2)+ 2\sqrt{QP}  (\rho_2 -  \rho_1 \rho_3)\Big)}{N + \Big( Q (1-{\rho_1}^2) +  P (1-{\rho_3}^2) +  2\sqrt{QP}  (\rho_2 -  \rho_1 \rho_3) \Big)}.
\end{align}

\subsection{Proof of Lemma \ref{lemma:OptimalRho1Rho3}}\label{sec:ProofLemmaOptimalRho}

We replace $\rho_2^{\star}$ in \eqref{eq:ConverseMinimizationCriteriaPb22} and we define \begin{align}
f({\rho_1}^2,{\rho_3}^2) = Q  (1-{\rho_1}^2 )+  P(1-{\rho_3}^2) - 2\sqrt{QP}   \sqrt{(1- {\rho_1}^2)(1  - {\rho_3}^2) -  \frac{N}{P}   {\rho_1}^2}.
\end{align}
Note that $f$ is well defined if ${\rho_1}^2\leq \frac{P}{P+N}$ and ${\rho_3}^2 \leq 1 -   \frac{N}{P}  \cdot \frac{{\rho_1}^2}{1- {\rho_1}^2}$. 
\begin{align}
&\frac{\partial f({\rho_1}^2,{\rho_3}^2)}{\partial {\rho_3}^2} = \sqrt{PQ}\cdot \frac{1-{\rho_1}^2}{\sqrt{(1- {\rho_1}^2)\cdot (1  - {\rho_3}^2) -  \frac{N}{P}  \cdot {\rho_1}^2}} - P,
\end{align}
then for all ${\rho_1}^2\leq \frac{P}{P+N}$, the optimal ${{\rho_3}^2}^{\star}({\rho_1}^2)$ is
\begin{align}
{{\rho_3}^2}^{\star}({\rho_1}^2) = \max\Bigg(1  -  \bigg( \frac{Q}{P}\cdot \Big(1-{\rho_1}^2\Big)+ \frac{N}{P}  \cdot \frac{{\rho_1}^2}{1- {\rho_1}^2}\bigg),0\Bigg).\label{eq:OptimalRho33}
\end{align}
We introduce the parameters
\begin{align}
\rho_a =&  \frac{2Q -(P+N) - \sqrt{(P+N)^2 - 4QN} }{2Q}, \label{eq:Condition2LemmaRho3}\\
\rho_b =&  \frac{2Q -(P+N) +\sqrt{(P+N)^2 - 4QN} }{2Q},\label{eq:Condition3LemmaRho3}
\end{align}
and we define the function
\begin{align}
F({\rho_1}^2)=f\Big({\rho_1}^2,{{\rho_3}^2}^{\star}({\rho_1}^2) \Big) 
=\begin{cases}
Q \cdot (1-{\rho_1}^2 )+  P  - 2\sqrt{QP} \cdot  \sqrt{1- {\rho_1}^2\cdot \frac{P+N}{P}}  &\text{ if }  0  \leq {\rho_1}^2\leq  \rho_a ,\\
N  \cdot \frac{{\rho_1}^2}{1- {\rho_1}^2}  & \text{ if }   \rho_a  \leq  {\rho_1}^2\leq\rho_b,\\
Q \cdot (1-{\rho_1}^2 )+  P  - 2\sqrt{QP} \cdot  \sqrt{1- {\rho_1}^2\cdot \frac{P+N}{P}} &\text{ if }  \rho_b \leq {\rho_1}^2\leq \frac{P}{P+N}.
\end{cases}
\end{align}
The function $F({\rho_1}^2)$ is continuous in  $ \rho_a$ and $\rho_b$. We define 
\begin{align}
\rho^{\star} =&  \frac{PQ - (P+N)^2}{Q(P+N)}.
\end{align}
$\bullet$ If $Q> 4N$ and $P\in[P_1,P_2]$, then the function $F({\rho_1}^2)$ is decreasing over the interval ${\rho_1}^2\in[0,\rho^{\star}]$ and increasing over the interval ${\rho_1}^2\in[\rho^{\star},\frac{P}{P+N}]$, then the optimal parameters are
\begin{align}
\rho_1= \sqrt{\rho^{\star}},\qquad 
\rho_2 = - \frac{P+N}{\sqrt{QP}},\qquad\rho_3= 0,
\end{align}
where $\rho_2$ is obtained from \eqref{eq:Rho2}.\\
$\bullet$ If $Q\leq4N$ or if $Q> 4N$ and $P\in[0,P_1]\cup[P_2,Q]$, then the optimal parameters are $\rho_1=\rho_3=0$ which imply $\rho_2=-1$.

\subsection{Proof of Lemma \ref{lemma:MultiplicativeEntropy}}\label{sec:lemma:MultiplicativeEntropy}

We consider $(X_0,W_2,U_1)\sim  \mc{N}(0,K)$ with covariance matrix
\begin{align}
K = 
\begin{pmatrix}
Q & \rho_1\sqrt{QV} & \rho_2\sqrt{QP} \\
\rho_1\sqrt{QV} &V & \rho_3\sqrt{VP} \\
 \rho_2\sqrt{QP} &  \rho_3\sqrt{VP} & P\\
\end{pmatrix},\label{eq:CovarianceMatrixApp}
\end{align}
where $(\rho_1,\rho_2,\rho_3)\in[-1,1]^3$ are such that $\det (K) =Q V  P \cdot \big(1- {\rho_1}^2- {\rho_2}^2- {\rho_3}^2 + 2 \rho_1 \rho_2\rho_3\big)\geq0$, i.e. $K$ is semi-definite positive.

We define $\tilde{W}_2 = \beta W_2$ with $\beta\in \R$. Then $(X_0,\tilde{W}_2,U_1)\sim  \mc{N}(0,\tilde{K})$ with covariance matrix
\begin{align}
\tilde{K} = 
\begin{pmatrix}
Q & \rho_1\sqrt{Q}(\beta\sqrt{V}) & \rho_2\sqrt{QP} \\
\rho_1\sqrt{Q}(\beta\sqrt{V}) &(\beta \sqrt{V})^2 & \rho_3\sqrt{P}(\beta\sqrt{V}) \\
 \rho_2\sqrt{QP} &  \rho_3\sqrt{P}(\beta\sqrt{V}) & P\\
\end{pmatrix},\label{eq:CovarianceMatrixApp}
\end{align}
and therefore $\det(\tilde{K}) = \beta^2 \det (K)$. Hence, we have
\begin{align}
H(X_0,\tilde{W}_2,U_1 ) =& \frac12\log_2\big((2\pi e)^3 \det ({K})\big)+ \log_2\big(\sqrt{\beta^2}\big) \\
=& H(X_0,{W}_2,U_1 ) + \log_2|\beta|,\\
H(X_0,\tilde{W}_2 ) =& H(X_0,{W}_2 ) + \log_2|\beta|.
\end{align}
This concludes the proof of Lemma \ref{lemma:MultiplicativeEntropy}.

\subsection{Proof of Lemma \ref{lemma:StateDependent}}\label{sec:lemma:StateDependent}

Consider the state-dependent channel 
\begin{align}
\tilde{Y_1} = \tilde{X}_0 + \tilde{W}_2 +  \tilde{U}_0 + \tilde{Z},
\end{align}
with Gaussian channel-state parameters $(\tilde{X}_0,\tilde{W}_2)\sim  \mc{N}(0,\tilde{K})$ and
\begin{align}
\tilde{K} = 
\begin{pmatrix}
q    &  \mu\sqrt{qv}\\
\mu\sqrt{qv} &v
\end{pmatrix},
\end{align} 
with $q\geq0$, $v\geq0$, $\mu\in[-1,1]$, and with Gaussian noise $\tilde{Z}\sim\mc{N}(0,N)$ such that $Z\perp( \tilde{X}_0, \tilde{W}_2, \tilde{U}_0)$. We consider that the channel input is also Gaussian $\tilde{U}_0\sim\mc{N}(0,P_0)$, $P_0\geq0$ with $\tilde{U}_0\perp(\tilde{X}_0,\tilde{W}_2)$ and we introduce Costa's auxiliary random variable, see \cite{Cwodp},
\begin{align}
 \tilde{W}_1 = \tilde{U}_0 + \alpha \tilde{X}_0, \qquad \alpha \in \R.
\end{align}

We have
\begin{align}
H(\tilde{W}_1|\tilde{X}_0,\tilde{W}_2)
=&  H(\tilde{U}_0 + \alpha   \tilde{X}_0|\tilde{X}_0,\tilde{W}_2)  
= H(\tilde{U}_0|\tilde{X}_0,\tilde{W}_2) = H(\tilde{U}_0) =  \frac{1}{2 }\log_2 \Big(2\pi e \cdot P_0\Big),\\
H(\tilde{Y_1},\tilde{W}_2) =&\frac{1}{2 }\log_2 \Big((2\pi e)^2 \cdot v  \big( q (1 - \mu^2) +  P_0 + N\big) \Big),\\
H(\tilde{W}_1,\tilde{Y}_1,\tilde{W}_2) =&\frac{1}{2 }\log_2 \bigg((2\pi e)^3 \cdot v   \Big(P_0 N  + q (1 - \mu^2) \big(P_0  (1-\alpha)^2  + N  \alpha^2\big)\Big)\bigg).
\end{align}
The details of the calculation are in App. \ref{sec:EvaluationEntropyCorrelatedState1} and \ref{sec:EvaluationEntropyCorrelatedState2}. We evaluate the information constraint
\begin{align}
I( \tilde{W}_1;\tilde{Y}_1,\tilde{W}_2) - I( \tilde{W}_1;\tilde{X}_0,\tilde{W}_2) 
=& H(\tilde{Y_1},\tilde{W}_2) - H(\tilde{W}_1,\tilde{Y}_1,\tilde{W}_2) + H(\tilde{W}_1|\tilde{X}_0,\tilde{W}_2)\\
=&\frac12\log_2 \Bigg(\frac{vP_0\big(q(1-\mu^2)+P_0+N\big)}{v\Big(P_0N+q(1-\mu^2)\big((1-\alpha)^2 P_0+ \alpha^2 N\big)\Big)}\Bigg)\\
=&\frac12\log_2 \Bigg(\frac{P_0\big(q(1-\mu^2)+P_0+N\big)}{P_0N+q(1-\mu^2)\big((1-\alpha)^2 P_0+ \alpha^2 N\big)}\Bigg).\label{eq:ICtildeBB}
\end{align}
\subsubsection{Evaluation of the entropy $H(\tilde{Y_1},\tilde{W}_2)$}\label{sec:EvaluationEntropyCorrelatedState1}

\begin{align}
\tilde{Y_1} =& \tilde{X}_0 + \tilde{W}_2 + \tilde{U}_0 + \tilde{Z}.
\end{align}
We have
\begin{align}
\E\Big[\tilde{W}_2^2\Big] =&v,\\
\E\Big[\tilde{Y}_1^2\Big]  =& \E\Big[(\tilde{X}_0 + \tilde{W}_2 + \tilde{U}_0 + \tilde{Z})^2\Big] = \E\Big[(\tilde{X}_0 + \tilde{W}_2)^2\Big]  + \E\Big[\tilde{U}_0^2\Big] + \E\Big[\tilde{Z}^2\Big] \\
=& \E\Big[\tilde{X}_0^2\Big]  + \E\Big[\tilde{W}_2^2\Big]   +2  \E\Big[\tilde{X}_0 \tilde{W}_2\Big]  + \E\Big[\tilde{U}_0^2\Big] + \E\Big[\tilde{Z}^2\Big] \\
=&q  + v   +2  \mu \sqrt{qv}  + P_0 + N\\
\E\Big[\tilde{W}_2 \tilde{Y}_1\Big] =&\E\Big[\tilde{W}_2(\tilde{X}_0 + \tilde{W}_2 + \tilde{U}_0 + \tilde{Z})\Big] =\E\Big[\tilde{W}_2\tilde{X}_0\Big]  + \E\Big[\tilde{W}_2^2\Big] \\
=& v + \mu \sqrt{qv}.
\end{align}

The Gaussian random variables $(\tilde{W}_2,\tilde{Y}_1)\sim \mathcal{N}(0,K)$ have covariance matrix 
\begin{align}
K = 
\begin{pmatrix}
v & v + \mu \sqrt{qv}\\
v + \mu \sqrt{qv} &q  + v   +2  \mu \sqrt{qv}  + P_0 + N\\
\end{pmatrix}.\label{eq:covarianceMatrixProof}
\end{align}

The determinant writes
\begin{align}
\det (K) = & v(q  + v   +2  \mu \sqrt{qv}  + P_0 + N) - (v + \mu \sqrt{qv})^2\\ 
= & v ( q  + v   +2  \mu \sqrt{qv}  + P_0 + N) - (v^2 + \mu^2 qv + 2v\mu \sqrt{qv})\\ 
= & v  \big( q  (1 - \mu^2) +  P_0 + N\big). 
\end{align}

Therefore,
\begin{align}
H(\tilde{Y}_1,\tilde{W}_2) =\frac{1}{2 }\log_2 \Big((2\pi e)^2 \cdot v \big( q (1 - \mu^2) +  P_0 + N\big) \Big).
\end{align}

\subsubsection{Evaluation of the entropy $H(\tilde{W}_1,\tilde{Y}_1,\tilde{W}_2)$}\label{sec:EvaluationEntropyCorrelatedState2}

\begin{align}
\tilde{Y}_1 =& \tilde{X}_0 + \tilde{W}_2 + \tilde{U}_0 + \tilde{Z},\\
\tilde{W}_1=& \tilde{U}_0 + \alpha   \tilde{X}_0,\quad \text{ with } \tilde{U}_0\perp (\tilde{X}_0,\tilde{W}_2),
\end{align}
We have
\begin{align}
\E\Big[\tilde{W}_1^2\Big] =& \E\Big[(\tilde{U}_0 + \alpha   X_0)^2\Big]
=  \E\Big[\tilde{U}_0^2\Big] + \E\Big[(\alpha   \tilde{X}_0 )^2\Big] = P_0 + \alpha^2   q  \\
\E\Big[\tilde{W}_1\tilde{W}_2\Big] =& \E\Big[(\tilde{U}_0 + \alpha   \tilde{X}_0 )\tilde{W}_2\Big] = \alpha    \E\Big[\tilde{X}_0 \tilde{W}_2\Big] =  \alpha    \mu \sqrt{qv},\\
\E\Big[\tilde{W}_1\tilde{Y}_1\Big] =& \E\Big[(\tilde{U}_0 + \alpha \tilde{X}_0 )\cdot(\tilde{X}_0 + \tilde{W}_2 + \tilde{U}_0 + \tilde{Z})\Big]\\
=& \E\Big[\tilde{U}_0 \Big] + \alpha  \E\Big[\tilde{X}_0^2\Big] + \alpha \E\Big[\tilde{X}_0  \tilde{W}_2\Big] = P_0 + \alpha  q + \alpha\mu \sqrt{qv} \\
\E\Big[\tilde{W}_2^2\Big] =&v,\\
\E\Big[\tilde{Y}_1^2\Big]  =& \E\Big[(\tilde{X}_0 + \tilde{W}_2 + \tilde{U}_0 + \tilde{Z})^2\Big] = \E\Big[(\tilde{X}_0 + \tilde{W}_2)^2\Big]  + \E\Big[\tilde{U}_0^2\Big] + \E\Big[\tilde{Z}^2\Big] \\
=& \E\Big[\tilde{X}_0^2\Big]  + \E\Big[\tilde{W}_2^2\Big]   +2  \E\Big[\tilde{X}_0 \tilde{W}_2\Big]  + \E\Big[\tilde{U}_0^2\Big] + \E\Big[\tilde{Z}^2\Big] \\
=&q  + v   +2  \mu \sqrt{qv}  + P_0 + N\\
\E\Big[\tilde{W}_2 \tilde{Y}_1\Big] =&\E\Big[\tilde{W}_2(\tilde{X}_0 + \tilde{W}_2 + \tilde{U}_0 + \tilde{Z})\Big] =\E\Big[\tilde{W}_2\tilde{X}_0\Big]  + \E\Big[\tilde{W}_2^2\Big] \\
=& v + \mu \sqrt{qv}.
\end{align}

The Gaussian random variables $(\tilde{W}_1,\tilde{W}_2,\tilde{Y}_1)\sim \mathcal{N}(0,K)$ have covariance matrix 
\begin{align}
K=
\begin{pmatrix}
P_0 + \alpha^2 q    &   \alpha     \mu \sqrt{qv}& P_0 + \alpha  q + \alpha  \mu \sqrt{qv}  \\
 \alpha     \mu \sqrt{qv}&v & v + \mu \sqrt{qv}\\
P_0 + \alpha  q + \alpha   \mu \sqrt{qv} &v + \mu \sqrt{qv} &q  + v   +2  \mu \sqrt{qv}  + P_0 + N\\
\end{pmatrix}.\label{eq:covarianceMatrixProof}
\end{align}

The determinant writes
\begin{align}
\det (K) = &\Big( P_0 + \alpha^2   q \Big)   v \big( q  (1 - \mu^2) +  P_0 + N\big)\\
-&   \alpha      \mu \sqrt{qv}  \cdot 
\det 
\begin{pmatrix}
  \alpha \mu \sqrt{qv} & P_0 + \alpha  q + \alpha   \mu \sqrt{qv}\\
v + \mu \sqrt{qv} &q  + v   +2  \mu \sqrt{qv}  + P_0 + N\\
\end{pmatrix}\\
+& \Big( P_0 + \alpha  q + \alpha    \mu \sqrt{qv}\Big) \cdot  \det
\begin{pmatrix}
 \alpha     \mu \sqrt{qv} & P_0 + \alpha  q + \alpha \mu \sqrt{qv} \\
v & v + \mu \sqrt{qv}\\
\end{pmatrix}\displaybreak[0]\label{eq:covarianceMatrixProof}\\
 = &\Big( P_0 + \alpha^2  q \Big) \cdot \big(  v  q \cdot(1 - \mu^2) +  v\cdot(P_0 + N)\big)\\
-& \alpha      \mu \sqrt{qv}\cdot \bigg( - \alpha vq(1-\mu^2) + N \alpha\mu \sqrt{vq} + P_0 \cdot \big( (\alpha-1)\mu \sqrt{vq} -v \big) \bigg)\\
+& \Big( P_0 + \alpha  q + \alpha   \mu \sqrt{qv}   \Big) \cdot \Big( - vP_0 - \alpha vq(1-\mu^2)\Big)\displaybreak[0]\label{eq:covarianceMatrixProof}\\
 = & P_0^2 v   - P_0^2 v +   P_0 \cdot \Bigg[N v + v  q (1 - \mu^2) + \alpha^2 v  q  -   \alpha      \mu \sqrt{qv} \big( (\alpha-1)\mu \sqrt{vq} - v \big)\nonumber\\
& -  \alpha vq(1-\mu^2) - v\Big(  \alpha  q + \alpha    \mu \sqrt{qv} \Big)\Bigg] + N \cdot \Bigg[ \alpha^2   q v - \big(\alpha \mu \sqrt{vq} \big)^2\Bigg]\nonumber\\
& +   \alpha ^2   q^2   v (1 - \mu^2) +  \alpha^2      \mu \sqrt{qv} vq(1-\mu^2)  - \big(  \alpha   q + \alpha     \mu \sqrt{qv} \big)  \alpha vq(1-\mu^2)\displaybreak[0]\label{eq:covarianceMatrixProof}\\
 = &  v P_0  \Big[N  + q (1 - \mu^2) (1-\alpha )^2 \Big] + vN  \Big[ q(1 - \mu^2) \alpha ^2 \Big]\\
  = &  v   \Big(P_0 N  + q (1 - \mu^2) \big(P_0  (1-\alpha )^2  + N  \alpha ^2\big)\Big).\label{eq:covarianceMatrixProof}
\end{align}

Hence we have
\begin{align}
H(\tilde{W}_1,\tilde{Y}_1,\tilde{W}_2) =&\frac{1}{2 }\log_2 \bigg((2\pi e)^3 \cdot v   \Big(P_0 N  + q (1 - \mu^2) \big(P_0  (1-\alpha)^2  + N  \alpha^2\big)\Big)\bigg).
\end{align}
This concludes the proof of Lemma \ref{lemma:StateDependent}.

\section{Proof of Proposition \ref{prop:IC_W2sign}}\label{sec:ProofPropW2discrete}

\subsection{Skew Gaussian random variables}
We consider $\rho\in [-1,1]$ and we define the channel input and the two auxiliary random variables by
\begin{align}
U_1 =& \rho \sqrt{\frac{P}{Q}} \cdot X_0 + \tilde{X},\qquad \text{ where } X_0\perp \tilde{X}\sim \mathcal{N}(0,P(1-\rho^2)),\label{eq:U1}\\
W_1 =& \tilde{X} + \frac{P(1-\rho^2)}{P(1-\rho^2)+N}\frac{(\sqrt{Q}+\rho \sqrt{P})}{\sqrt{Q}} \cdot X_0,\label{eq:W1_Costa_W2discrete}\\
W_2 =& \sign(X_1)\in \{-1,1\}.\label{eq:W2}
\end{align}
Equation \eqref{eq:W1_Costa_W2discrete} corresponds to Costa's optimal auxiliary random variable $W_1$, see \cite{Cwodp}, with parameter $\alpha = \frac{P(1-\rho^2)}{(P(1-\rho^2)+N)}$ for the channel state $\frac{(\sqrt{Q}+\rho \sqrt{P})}{\sqrt{Q}} \cdot X_0 \sim \mc{N}(0,(\sqrt{Q}+\rho \sqrt{P})^2)$ and the channel input constraint $\E[\tilde{X}^2]\leq P(1-\rho^2)$. 

\begin{lemma}\label{lemma:Entropy_W2_discrete}
Suppose that $(U_1,W_1,W_2)$ are defined according to \eqref{eq:U1}-\eqref{eq:W2}, for some $\rho\in [-1,1]$. Then
\begin{align}
H(X_0,W_1|W_2)=& \frac12 \log_2\Big((2\pi e)^2\cdot PQ(1-\rho^2)\Big) -1,\\
H(Y|W_2) =& \frac12 \log_2\Bigg((2\pi e)\cdot (T+N)\Bigg) - \Psi\Bigg(\sqrt{\frac{T}{N}}\Bigg),\\
H(Y,W_1|W_2)=& \frac12 \log_2\Bigg((2\pi e)^2\cdot \frac{(T+N)NP(1-\rho^2)}{P(1-\rho^2)+N}\Bigg) - \Psi\Bigg( \sqrt{\frac{T(\sqrt{Q}+\rho\sqrt{P})^2N+P(1-\rho^2)(T+N)^2}{(\sqrt{Q}+\rho\sqrt{P})^2N^2}}\Bigg),
\end{align}
where the entropy reduction function $\Psi:\R \to [0,1]$ is defined by 
\begin{align}
\Psi(\alpha) =&\int2\Phi\big(\alpha \cdot x\big)\cdot\log_2\Big(2\Phi\big(\alpha \cdot x\big)\Big) \frac{1}{\sqrt{2\pi}}\exp\Big(-\frac{x^2}{2}\Big)dx,
\end{align}
and $\Phi(x) = \frac{1}{\sqrt{2\pi}}\int_{-\infty}^{x} \exp\big(-\frac{t^2}{2}\big) dt $, is the Gaussian cumulative distribution function.
\end{lemma}

The proof of Lemma \ref{lemma:Entropy_W2_discrete} is stated in App.~\ref{sec:proof:lemma:Entropy_W2_discrete}. According to Lemma \ref{lemma:Entropy_W2_discrete}, the information constraint writes
\begin{align}
&I(W_1;Y,W_2) - I(W_1;X_0,W_2) - I(X_0;W_2) \\
=& H(X_0,W_1|W_2) - H(Y,W_1|W_2)+H(Y|W_2)  - H(X_0)\\
=&  \frac12 \log_2\Bigg(1+\frac{P(1-\rho^2)}{N}\Bigg) -1  -  \Psi\Bigg(\sqrt{\frac{T}{N}}\Bigg)+ \Psi\Bigg( \sqrt{\frac{T(\sqrt{Q}+\rho\sqrt{P})^2N+P(1-\rho^2)(T+N)^2}{(\sqrt{Q}+\rho\sqrt{P})^2N^2}}\Bigg).
\end{align}

The random variable $X_1\sim \mc{N}(0,T)$ is Gaussian centred, thus $\mathbb{P}\{X_1\geq0\}=\frac12$. The probability density function of the skew Gaussian distributions writes
\begin{align}
f(y_1|X_1\geq0) =& \frac{2}{\sqrt{T+N}} \cdot \Phi\bigg( y_1 \cdot \sqrt{\frac{ T}{N( T+N)}} \bigg) \cdot \phi\bigg( \frac{y_1}{\sqrt{T+N}}\bigg),\qquad \forall y_1\in \R,\\
f(x_1|y_1,X_1\geq0) =&\frac{1}{\sqrt{ \frac{T N}{T+N}}}\frac{\phi\Big( \frac{x_1-y_1 \frac{T}{T+N} }{\sqrt{ \frac{T N}{T+N}}} \Big)}{\Phi\Big(y_1 \cdot \sqrt{\frac{ T}{N( T+N)}}\Big)},\qquad\forall x_1\geq0, \; \forall y_1\in \R.
 \end{align}
The conditional variance of a skew Gaussian distribution writes
\begin{align}
&  \E\Big[\Big(X_1 - \E[X_1|Y_1=y_1,X_1\geq0]\Big)^2\Big|Y_1=y_1,X_1\geq0\Big]\\
=& \frac{T N}{T + N} \cdot \Bigg(1 - \frac{ y_1 \cdot \sqrt{\frac{ T}{N( T+N)}}\phi\Big( y_1 \cdot \sqrt{\frac{ T}{N( T+N)}}\Big) }{\Phi\Big( y_1 \cdot \sqrt{\frac{ T}{N( T+N)}}\Big)} - \bigg(\frac{\phi\Big( y_1 \cdot \sqrt{\frac{ T}{N( T+N)}}\Big) }{  \Phi\Big(y_1 \cdot \sqrt{\frac{ T}{N( T+N)}}\Big)}\bigg)^2\Bigg).
\end{align}
By symmetry we have
\begin{align}
&\E\Big[\Big(X_1 - \E[X_1|W_2,Y_1]\Big)^2\Big]\\
=& \mathbb{P}\{W_2=1\}\cdot \E\Big[\Big(X_1 - \E[X_1|W_2,Y_1=1]\Big)^2\Big|W_2=1\Big]\\ 
&+ \mathbb{P}\{W_2=-1\}\cdot\E\Big[\Big(X_1 - \E[X_1|W_2,Y_1=-1]\Big)^2\Big|W_2=-1\Big] \\
=& \E\Big[\Big(X_1 - \E[X_1|Y_1,X_1\geq0]\Big)^2\Big|X_1\geq0\Big]\\ 
=&\int \E\Big[\Big(X_1 - \E[X_1|Y_1=y_1,X_1\geq0]\Big)^2\Big|Y_1=y_1,X_1\geq0\Big]f(y_1|X_1\geq0) dy_1\\ 
=& \int \frac{T N}{T + N} \cdot \Bigg(1 - \frac{ y_1 \cdot \sqrt{\frac{ T}{N( T+N)}}\phi\Big( y_1 \cdot \sqrt{\frac{ T}{N( T+N)}}\Big) }{\Phi\Big( y_1 \cdot \sqrt{\frac{ T}{N( T+N)}}\Big)} - \bigg(\frac{\phi\Big( y_1 \cdot \sqrt{\frac{ T}{N( T+N)}}\Big) }{  \Phi\Big(y_1 \cdot \sqrt{\frac{ T}{N( T+N)}}\Big)}\bigg)^2\Bigg) \\
&\times \frac{2}{\sqrt{T+N}} \cdot \Phi\bigg(y_1 \cdot \sqrt{\frac{ T}{N( T+N)}}\bigg) \cdot \phi\bigg( \frac{y_1}{\sqrt{T+N}}\bigg)  dy_1\\
=& \frac{T N}{T + N} \cdot \Bigg(1 -  \frac{2}{\sqrt{T+N}} \cdot \frac{1}{2\pi} \int\frac{  \phi\Big( y_1 \cdot \sqrt{\frac{2 T+N}{N( T+N)}}\Big) }{  \Phi\Big(y_1 \cdot \sqrt{\frac{T}{N( T+N)}}\Big)}  dy_1  \Bigg),
\end{align}
where the first integral is equal to zero. This concludes the proof of Proposition \ref{prop:IC_W2sign}.


\subsection{Proof of Lemma \ref{lemma:Entropy_W2_discrete}}\label{sec:proof:lemma:Entropy_W2_discrete}

\subsubsection{Evaluation of $H(X_0,W_1|W_2)$}\label{sec:EvaluationEntropy_W2_discreteHX0W1W2}

The interim state $X_1$ is a linear combination of $W_1$ and $X_0$,
\begin{align}
X_1 =& U_1+ X_0 = \tilde{X} + \frac{(\sqrt{Q}+\rho \sqrt{P})}{\sqrt{Q}} \cdot X_0 = W_1 + \frac{N}{(P(1-\rho^2)+N)}\frac{(\sqrt{Q}+\rho \sqrt{P})}{\sqrt{Q}} \cdot X_0.
\end{align}
Since $W_2=\sign(X_1)$, we have $H(W_2|X_0,W_1)=0$ and $H(W_2)=1$.
\begin{align}
H(X_0,W_1|W_2) 
=& H(X_0,W_1) + H(W_2|X_0,W_1) - H(W_2)\\
=& H(X_0,W_1)-1\\
=& \frac12 \log_2\Big((2\pi e)^2\cdot PQ(1-\rho^2)\Big) -1.
\end{align}
Indeed, the determinant of the covariance matrix  of $(X_0,W_1)\sim\mathcal{N}(0,K_{X_0W_1})$ satisfies 
\begin{align}
\det(K_{X_0W_1})
=
\begin{vmatrix}
Q & \frac{P(1-\rho^2)}{P(1-\rho^2)+N}\frac{(\sqrt{Q}+\rho \sqrt{P})}{\sqrt{Q}} Q\\
\frac{P(1-\rho^2)}{P(1-\rho^2)+N}\frac{(\sqrt{Q}+\rho \sqrt{P})}{\sqrt{Q}} Q &
P(1-\rho^2)+\bigg(\frac{P(1-\rho^2)}{P(1-\rho^2)+N}\frac{(\sqrt{Q}+\rho \sqrt{P})}{\sqrt{Q}}\bigg)^2 Q\\
\end{vmatrix} = PQ(1-\rho^2).\label{eq:CovarianceMatrixKX0W1}
\end{align}

\subsubsection{Evaluation of $H(Y_1,W_1|W_2)$ and $H(Y_1|W_2)$}\label{sec:EvaluationEntropy_W2_discreteHY1W1W2_HY1W2}

By using the change of variable $T = P + Q+ 2\rho \sqrt{PQ}$, the covariance matrix $K_{X_1Y_1W_1}$ of the random variables $(X_1,Y_1,W_1)\sim\mathcal{N}(0,K_{X_1Y_1W_1})$ is given by
\begin{align}
K_{X_1Y_1W_1} = 
\begin{pmatrix}
T & T &  \frac{P(1-\rho^2)(T+N)}{P(1-\rho^2)+N}\\
T & T+N & \frac{P(1-\rho^2)(T+N)}{P(1-\rho^2)+N} \\
 \frac{P(1-\rho^2)(T+N)}{P(1-\rho^2)+N} &  \frac{P(1-\rho^2)(T+N)}{P(1-\rho^2)+N}&P(1-\rho^2)+ \frac{(P(1-\rho^2))^2(\sqrt{Q}+\rho \sqrt{P})^2}{(P(1-\rho^2)+N)^2}\\
\end{pmatrix}.\label{eq:CovarianceMatrixX1Y1W1}
\end{align}

Knowing that $W_2=1$ $\Leftrightarrow$ $X_1\geq0$, the random variables $(Y_1,W_1)\sim \mathcal{SN}(0,K_{Y_1W_1},\delta_{Y_1W_1})$ are bi-variate skew Gaussian random variables with skewness value and determinant of covariance matrix $K_{Y_1W_1}$ of $(Y_1,W_1)\sim\mathcal{N}(0,K_{Y_1W_1})$ given by
\begin{align}
\delta_{Y_1}=& \sqrt{\frac{T}{N}},\\
\delta_{Y_1W_1}=&  \sqrt{\frac{T (\sqrt{Q}+\rho \sqrt{P})^2N + P(1-\rho^2)(T +N)^2}{(\sqrt{Q}+\rho \sqrt{P})^2N^2}},\\
\det(K_{Y_1W_1}) =& \frac{(T+N)NP(1-\rho^2)}{P(1-\rho^2)+N}.
\end{align}
According to \cite[Prop. 3, pp 49]{ACG13semi}, the conditional entropy writes 
\begin{align}
H(Y_1,W_1|W_2=1) =& \frac12 \log_2\Bigg((2\pi e)^2\cdot \frac{(T+N)NP(1-\rho^2)}{P(1-\rho^2)+N}\Bigg) \nonumber\\
&- \Psi\Bigg( \sqrt{\frac{T(\sqrt{Q}+\rho\sqrt{P})^2N+P(1-\rho^2)(T+N)^2}{(\sqrt{Q}+\rho\sqrt{P})^2N^2}}\Bigg),\\
H(Y_1|W_2=1) =& \frac12 \log_2\Bigg((2\pi e)\cdot (T+N)\Bigg) - \Psi\Bigg(\sqrt{\frac{T}{N}}\Bigg).
\end{align}
 By symmetry, we obtain the same expressions for $H(Y_1,W_1|W_2=-1)$ and $H(Y_1|W_2=-1)$, thus also for $H(Y_1,W_1|W_2)$ and $H(Y_1|W_2)$.

\section{Proof of Prop. \ref{prop:TwoPointStrategy}}\label{sec:ProofPropTwoPoint}

\subsection{Power Cost}
Depending on the parameter $a\geq0$, the power cost is given by
\begin{align}
P_{\mathsf{two}}(a) = &\E \big[ U_1^2\big] = \E \Big[ \Big(a \cdot \textrm{sign}\big(X_0\big) - X_0 \Big)^2\Big]\\
=& a^2 + Q - 2 a \E \Big[|X_0|\Big] \\
=&Q + a\Big(a-2\sqrt{\frac{2Q}{\pi}}\Big).
\end{align}

\subsection{MMSE Cost}

The random variable $X_1\in \{-a,a\}$ is uniformly distributed and the channel is $Y_1=X_1+Z$ where $X_1\perp Z\sim\mc{N}(0,N)$. Therefore, the joint PDF of $(X_1,Y_1)$ writes for all $(x_1,y_1)\in \R^2$
\begin{align}
f(x_1,y_1) =& \frac12 \frac{1}{\sqrt{N}} \phi\bigg(\frac{y_1 - x_1}{\sqrt{N}}\bigg).
\end{align}
Moreover,
\begin{align}
f(x_1|y) =&\frac{f(x_1,y)}{f(a,y)+f(-a,y)}=  \frac{\phi\big(\frac{y - x_1}{\sqrt{N}}\big)}{\phi\big(\frac{y - a}{\sqrt{N}}\big)+ \phi\big(\frac{y +a}{\sqrt{N}}\big)},\\
f(y_1) =&  \frac12 \frac{1}{\sqrt{N}} \phi\bigg(\frac{y - a}{\sqrt{N}}\bigg)+  \frac12 \frac{1}{\sqrt{N}} \phi\bigg(\frac{y +a}{\sqrt{N}}\bigg)\\
=&  \frac{1}{\sqrt{N}}  \frac{1}{\sqrt{2 \pi}}\exp\bigg(-\frac{y_1^2+a^2}{2N}\bigg)\Bigg( \frac12\exp\bigg(\frac{ay_1}{N} \bigg)+ \frac12\exp\bigg(-\frac{ay_1}{N} \bigg)\Bigg)\\
=&  \sqrt{\frac{2\pi}{N}}  \phi\bigg(\frac{a}{\sqrt{N}}\bigg)\phi\bigg(\frac{y_1}{\sqrt{N}}\bigg) \cosh\bigg(\frac{ay_1}{N} \bigg).
\end{align}
and
\begin{align}
\E[X_1|Y_1=y_1] =& a \frac{\phi\big(\frac{y_1 - a}{\sqrt{N}}\big)- \phi\big(\frac{y_1 +a}{\sqrt{N}}\big)}{\phi\big(\frac{y_1 - a}{\sqrt{N}}\big)+ \phi\big(\frac{y_1 +a}{\sqrt{N}}\big)}\\
=& a \frac{\exp\Big(\frac{a y_1 }{N}\Big)- \exp\Big(-\frac{a y_1 }{N}\Big)}{\exp\Big(\frac{a y_1 }{N}\Big)+ \exp\Big(-\frac{a y_1 }{N}\Big)}\\
=&a \tanh\bigg(\frac{a y_1 }{N}\bigg).
\end{align}
Therefore, 
\begin{align}
&\E\Big[\big(X_1-\E[X_1|Y_1]\big)^2\Big] \\
=& \int f(x_1,y_1)\bigg(x_1 - a \tanh\bigg(\frac{a y_1 }{N}\bigg)\bigg)^2 dy_1dx_1\\
=& a^2+ a^2\int f(y_1)\bigg(\tanh\Big(\frac{a y_1 }{N}\Big)\bigg)^2dy_1
- 2a \int f(y_1)  \tanh\bigg(\frac{a y_1 }{N}\bigg)\E[X_1|Y_1=y_1] dy_1\\
=& a^2- a^2\int f(y_1) \bigg(\tanh\Big(\frac{a y_1 }{N}\Big)\bigg)^2dy_1\\
=&a^2\int f(y_1) \bigg( \frac{1}{\cosh\big(\frac{a y_1 }{N}\big)}\bigg)^2dy_1\\
=&a^2\sqrt{\frac{2\pi}{N}}  \phi\bigg(\frac{a}{\sqrt{N}}\bigg)\int  \frac{\phi\big(\frac{y_1}{\sqrt{N}}\big)}{\cosh\big(\frac{a y_1 }{N}\big)}dy_1.
\end{align}

\section{Proof of Prop. \ref{prop:DPC}}\label{sec:ProofPropDPC}

We define the function $g(\alpha)$ and we compute its derivative $g'(\alpha)$.
\begin{align}
g(\alpha) = &\frac{NPQ(1-\alpha)^2 }{PQ(1-\alpha)^2 + N(P+\alpha^2 Q)},\\
g'(\alpha) =&\frac{N^2 PQ 2(\alpha -1) (P + \alpha Q) }{\big(PQ(1-\alpha)^2 + N(P+\alpha^2 Q)\big)^2} .
\end{align}
Therefore, the function $g(\alpha)$ is increasing over the interval $]-\infty,-\frac{P}{Q}]$, decreasing over the interval $[-\frac{P}{Q},1]$ where it is equal to zero, and increasing over the interval $[1,+\infty[$. Therefore, the solution to the optimization problem \eqref{eq:Optimization_DPC} is the largest $\alpha^{\star}\in [0,1]$ such that 
\begin{align}
&P(P+Q+N)-  PQ(1-\alpha)^2 - N(P+\alpha^2 Q) \geq 0 \label{eq:trinome_alpha}\\
\Longleftrightarrow &P^2 + 2 \alpha PQ  - \alpha^2 Q (P+N ) \geq 0.
\end{align}
This equation has two solutions,
\begin{align}
\alpha_1 &= \frac{-2PQ-2P\sqrt{Q(P+Q+N)}}{-2Q(P+N)}= \frac{P(\sqrt{Q}+\sqrt{P+Q+N})}{\sqrt{Q}(P+N)}>0,\\
\alpha_2 &= \frac{-2PQ+2P\sqrt{Q(P+Q+N)}}{-2Q(P+N)}= \frac{P(\sqrt{Q}-\sqrt{P+Q+N})}{\sqrt{Q}(P+N)}<0.
\end{align}
The optimal solution to the problem \eqref{eq:Optimization_DPC} is 
\begin{align}
\alpha^{\star} = \min\Bigg(1,\frac{P(\sqrt{Q}+\sqrt{P+Q+N})}{\sqrt{Q}(P+N)}\Bigg).
\end{align}
Note that 
\begin{align}
\frac{P(\sqrt{Q}+\sqrt{P+Q+N})}{\sqrt{Q}(P+N)} \geq 1 
\Longleftrightarrow  P^2(P+Q+N)  \geq Q N^2.\label{eq:function_P}
\end{align}
We denote by $P^{\star}$ the unique positive solution of $P^2(P+Q+N) - Q N^2=0$. Note that if $P> P^{\star}$, then $\alpha^{\star}=1$ and the receiver retrieves $W = U_1 + X_0 = X_1$, which implies $\mathsf{MMSE}_{\mathsf{dpc}}(P) =0$. Suppose that $P\leq P^{\star}$, then
\begin{align}
\mathsf{MMSE}_{\mathsf{dpc}}(P)  = & g(\alpha_1) = \frac{NPQ(1-\alpha_1)^2 }{PQ(1-\alpha_1)^2 + N(P+\alpha_1^2 Q)}\\
=& \frac{N\big( N\sqrt{Q} - P\sqrt{P+Q+N}\big)^2}{(P+N)^2(P+Q+N)}.
\end{align}

\begin{remark}
In \cite[App.~D.7, Eq. (72)]{GS10wcaa}, the optimization problem of \eqref{eq:Optimization_DPC} is performed in function of parameter $P\geq0$, for each $\alpha\in \R$. In particular, for all $\alpha\in \R$, the parameter $P\geq0$  that satisfies  $P(P+Q+N)= PQ(1-\alpha)^2 + N(P+\alpha^2 Q)$, is given by
\begin{align}
P^{\circ} = &\frac{Q\alpha(2-\alpha)}{2}\bigg(\sqrt{1+\frac{4N}{Q(2-\alpha)^2}}-1\bigg).
\end{align}
We believe there is a typo in \cite[Eq. (71)]{GS10wcaa}, where 
\begin{align}
P = \frac{\sqrt{Q\alpha(2-\alpha)}}{2}\bigg(\sqrt{1+\frac{4N}{Q(2-\alpha)^2}}-1\bigg).
\end{align}
\end{remark}


\end{document}